\documentclass[11pt]{article}
\usepackage{amsmath}
\usepackage{amssymb}
\usepackage{amsthm}
\usepackage{graphicx}
\usepackage[margin=1in]{geometry}
\usepackage{float}
\usepackage{algorithm}
\usepackage{algpseudocode}
\usepackage{hyperref}
\usepackage{listings}
\usepackage{pdflscape}
\usepackage{setspace}
\usepackage{longtable}
\usepackage{young}
\usepackage{tikz-cd}
\usepackage{thmtools}
\usepackage[affil-it]{authblk}
\usepackage[english]{babel}
\usepackage{blindtext}
\usepackage{multicol}
\usepackage{authblk}

\tikzstyle{startstop} = [rectangle, rounded corners, minimum width=1cm, minimum height=2cm,text centered, draw=black, fill=gray!30]
\tikzstyle{arrow} = [thick,->,>=stealth]

\allowdisplaybreaks

\newtheorem{mytheo}{Theorem}
\newtheorem{mylem}{Lemma}

\newtheorem{mycor}{Corollary}
\newtheorem{mydef}{Definition}
\newtheorem{myrem}{Remark}
\newtheorem{myexamp}{Example}

\newcommand{\norm}[1]{\left\lVert#1\right\rVert}

\title{Quantum Fourier Sampling is  Guaranteed to Fail to Compute Automorphism Groups of Easy Graphs}

\author{Omar Shehab\thanks{\texttt{shehab1@umbc.edu}. Corresponding author.}\hspace{0.25cm}}
\author{Samuel J. Lomonaco Jr.\thanks{\texttt{lomonaco@umbc.edu}}}
\affil{Department of Computer Science and Electrical Engineering

University of Maryland Baltimore County, Baltimore, MD}
\date{}

\begin{document}
\maketitle

\linespread{1.0}
\begin{abstract}
The quantum hidden subgroup approach is an actively studied approach to solve combinatorial problems in quantum complexity theory. With the success of the Shor's algorithm, it was hoped that similar approach may be useful to solve the other combinatorial problems. One such problem is the graph isomorphism problem which has survived decades of efforts using the hidden subgroup approach. This paper provides a systematic approach to create arbitrarily large classes of classically efficiently  solvable  graph automorphism problems or {\it easy} graph automorphism problems for which hidden subgroup approach is guaranteed to always fail irrespective of the size of the graphs no matter how many copies of coset states are used. As the graph isomorphism problem is believed to be at least as hard as the graph automorphism problem, the result of this paper entails that the hidden subgroup approach is also guaranteed to always fail for the arbitrarily large classes of graph isomorphism problems. Combining these two results, it is argued that the hidden subgroup approach is essentially a dead end and alternative quantum algorithmic approach needs to be investigated for the graph isomorphism and automorphism problems.
\end{abstract}
\newpage
\tableofcontents

\section{Introduction}
\label{sec:auto-hsp}
The main result of this paper is the Theorem ~\ref{theo:cycle-graph-auto-weak-fails} and the Corollary ~\ref{cor:cycle-graph-auto-strong-fails} where it is shown that both the weak and strong quantum Fourier samplings are guaranteed to always fail for the classically trivial cycle graph automorphism problem. It is also shown how to systematically determine the non-trivial classes of graphs for which quantum Fourier transform (QFT) always fails to construct the automorphism groups. Here, the term `non-trivial' refers to the classes of graph automorphism problems which can be solved trivially on a classical computer. This result puts an end to the decades long effort of finding a hidden subgroup algorithm for the graph isomorphism problem. Previously, the researchers have been proving increasing negative results which indicated that the probability of successfully deciding a graph isomorphism problem is exponentially small. This paper gives an algorithm to create arbitrarily large classes of graph automorphism problems for which the quantum hidden subgroup approach will always fail to compute the automorphism group no matter how big the computer or how small the size of the problem is. As the graph automorphism problem is Karp-reducible \cite{kobler2012graph} to the graph isomorphism problem, it can be inferred that there are classes of graph isomorphism problem for which hidden subgroup approach will always fail irrespective of the size of the problem. The linear representation theory of the dihedral groups plays a very important role in proving this result.

The general framework for the hidden subgroup problems was first formulated in \cite{brassard1997exact, hoyer2000quantum, mosca1999hidden}. The hidden subgroup problem can be defined as follows \footnote{This pedagogically convenient version of HSP has been borrowed from the presentation titled 'Graph isomorphism, the hidden subgroup problem and identifying quantum states' by Pranab Sen.}.

\begin{mydef}[Hidden subgroup problem]

{\bf Given:} $G:$ group, $S:$ set, $f : G \to S$ via an oracle.\\
{\bf Promise:} Subgroup $H \le G$ such that $f$ is constant on the left
cosets of $H$ and distinct on different cosets.\\
{\bf Task:} Find the hidden subgroup $H$ by querying $f$.
\end{mydef}

The hidden subgroup version of the graph isomorphism problem was first defined in \cite{jozsa1998quantum}.The $n$-vertex graph isomorphism problem for rigid graphs of  $n$ vertices can be expressed as a hidden subgroup problem over the ambient symmetric group $S_{2n}$ or more specifically the wreath product $S_n \wr \mathbb{Z}_2$ where the hidden subgroup is promised to be either trivial or of order two \cite{moore2010impossibility}.  The scheme and notation of the following definitions of the hidden subgroup problem, used in this paper, are borrowed from \cite{lomonacopers2002, van2012quantum}. Erd{\H{o}}s et al \cite{erdHos1963asymmetric} have shown that the automorphism groups of the most of the graphs are trivial.  So, although, the problem was defined for all simple undirected graphs in \cite{jozsa1998quantum}, this paper follows the example of \cite{grigni2001quantum} and limit the discussion to the rigid graphs with trivial automorphism groups. The rigidity of the graphs in this definition will be temporarily relaxed in Section ~\ref{sec:iso-auto} to prove reducibility.

\begin{mydef} [Graph isomorphism as a hidden subgroup problem ($\text{{\bf GI}}_{\text{HSP}}$)]
Let the $2 n$ vertex graph $\Gamma = \Gamma_1 \sqcup \Gamma_2$ be the disjoint union of the two rigid graphs $\Gamma_1$ and $\Gamma_2$ such that $Aut \left(\Gamma_1\right) = Aut \left(\Gamma_2\right) = \left\{e\right\}$. A map $\varphi : S_{2n} \to \text{Mat}\left(\mathbb{C}, N \right)$ \footnote{$\text{Mat}\left(\mathbb{C}, N \right)$ is the algebra of all $N \times N$ matrices over the complex numbers $\mathbb{C}$.} from the group $S_{2n}$ is said to have hidden subgroup structure if there exists a subgroup $H_\varphi$ of $S_{2n}$, called a hidden subgroup, an injection $\ell_\varphi : S_{2n}/H \to \text{Mat}\left(\mathbb{C}, N \right)$, called a hidden injection, such that the diagram
\[
\begin{tikzcd}
S_{2n} \arrow{r}{\varphi} \arrow{d}{\nu} & \text{Mat}\left(\mathbb{C}, N \right)\\
S_{2n}/H \arrow{ur}{\ell_\varphi}&
\end{tikzcd}
\]
is a commutative diagram, where $S_{2n}/H_{\varphi}$ denotes the collection of right cosets of $H_\varphi$ in $S_{2n}$, and where $\nu : S_{2n}/H_\varphi$ is the natural map of $S_{2n}$ onto $S_{2n}/H_\varphi$. The group $S_{2n}$ is called the ambient group and  the set $\text{Mat}\left(\mathbb{C}, N \right)$ is called the target set.

The hidden subgroup version of the graph isomorphism problem is to determine a hidden subgroup $H$ of $S_{2n}$ with the promise that $H$ is either trivial or $|H| = 2$.
\end{mydef}

This section also gives a formal definition for the hidden subgroup representation of the graph automorphism problem.

\begin{mydef} [Graph automorphism as a hidden subgroup problem ($\text{{\bf GA}}_{\text{HSP}}$)]
For a graph $\Gamma$ with $n$ vertices, a map $\varphi : S_{n} \to \text{Mat}\left(\mathbb{C}, N \right)$ \footnote{$\text{Mat}\left(\mathbb{C}, N \right)$ is the algebra of all $N \times N$ matrices over the complex numbers $\mathbb{C}$.} from the group $S_{n}$ is said to have hidden subgroup structure if there exists a subgroup $\text{Aut}\left(\Gamma\right)$ of $S_{n}$, called a hidden subgroup, an injection $\ell_\varphi : S_{n}/\text{Aut}\left(\Gamma\right) \to \text{Mat}\left(\mathbb{C}, N \right)$, called a hidden injection, such that for each $g \in \text{Aut}\left(\Gamma\right)$, $g \left(\Gamma\right) = \Gamma$ and, the diagram
\[
\begin{tikzcd}
S_{n} \arrow{r}{\varphi} \arrow{d}{\nu} & \text{Mat}\left(\mathbb{C}, N \right)\\
S_{n}/\text{Aut}\left(\Gamma\right) \arrow{ur}{\ell_\varphi}&
\end{tikzcd}
\]
is commutative, where $S_{n}/\text{Aut}\left(\Gamma\right)$ denotes the collection of right cosets of $\text{Aut}\left(\Gamma\right)$ in $S_{n}$, and where $\nu : S_{n}/\text{Aut}\left(\Gamma\right)$ is the natural map of $S_{n}$ onto $S_{n}/\text{Aut}\left(\Gamma\right)$.  $S_{n}$ is called  the ambient group and  $\text{Mat}\left(\mathbb{C}, N \right)$ is called  the target set.

The hidden subgroup version of the graph automorphism problem is to determine a hidden subgroup $\text{Aut}\left(\Gamma\right)$ of $S_{n}$ with the promise that $\text{Aut}\left(\Gamma\right)$ is either of trivial or non-trivial order depending on the type of $\Gamma$.
\end{mydef}

\subsection{Outline of the paper}
For the convenience of the readers, a brief outline of the paper is given here. The Section ~\ref{sec:Historical-context} gives an overview of the related literature, Section  ~\ref{sec:prem} provides the preliminary background needed to follow the discussion used in this paper, Section ~\ref{sec:qft-gi} provides the already known results on the hidden subgroup approach for graph isomorphism, Section ~\ref{sec:weak-cycle} presents the original result that the hidden subgroup approach is guaranteed to fail for an easy class of graph automorphism problem, and Section ~\ref{sec:qft-ga-plus} presents another original result which is a systematic way to build arbitrarily  large classes of graph automorphism problems for which hidden subgroup algorithms are guaranteed to fail. Finally, it has been discussed whether the hidden subgroup algorithms are the most appropriate ways to attempt combinatorial problems in quantum computation.

\subsection{Key technical ideas}
This section summarizes the key technical ideas used to prove the main results of this paper.  The automorphism group of a cycle graph is the dihedral group $D_n$ of order $2 n$. It has been shown in  Lemma ~\ref{lem:d-n-1-d-rep-prob-zero} that the probability of measuring the labels of one dimensional irreducible representations of $D_n$ is zero for non-trivial representations. Then, in Lemma ~\ref{lem:d-n-2-d-rep-prob-zero},  it has been shown that the probability of measuring the labels of two dimensional irreducible representations of $D_n$ is always zero. Combining these two lemmas, it has been proved in Theorem ~\ref{theo:cycle-graph-auto-weak-fails} that Weak quantum Fourier sampling always fails to solve the cycle graph automorphism problem irrespective of its size. Finally, this paper gives Algorithm ~\ref{algo:arb-fail} to create arbitrarily large class of graph automorphism problems (with rotational symmetries) for which quantum Fourier transform is guaranteed to always fail irrespective of the size of the graphs.

It was already proved in \cite{grigni2001quantum} that single coset state can provide only exponentially less information for the graph isomorphism problem. Hence, later works, for example  \cite{hallgren2010limitations, moore2010impossibility}, investigated the possibility of using multiple copies of coset states. The current work gives a stronger result on the graph automorphism problem for the single coset state. Moreover, all the multi-coset state algorithms are conditioned on the successful execution of weak sampling, which has been proved in this paper to fail with guarantee for the problem of interest. So, the case of failure can also be inferred for multi-coset approaches, e.g., the sieve algorithms. To summarize, if there are rotational symmetries in the graph, we will not get exponentially less information with a single copy coset state rather, but, even worse, we will get exactly zero amount of information no matter how large the quantum computer is. Same would be true for k-copy coset states.

\section{Historical context}
\label{sec:Historical-context}
Read et al. \cite{Read1977} have named the tendency of incessant but unsuccessful attempts at the graph isomorphism problem as the {\it graph isomorphism disease}. This indicates the amount of interest about the problem among the researchers. For almost three decades, until $2015$, the  best known algorithm for the general graph isomorphism problem has been due to Babai et al. \cite{babai1983canonical}. The algorithm exploits graph canonization techniques through label reordering in exponential time (\ensuremath{exp \left(n^{\frac{1}{2} + o\left(1\right)}\right)}), where \ensuremath{n = |V|}. Faster algorithms have been proposed for graph sub classes with special properties. In \cite{babai1983canonical}, Babai et al. also proved the bound for tournament graphs is \ensuremath{n^{\left(\frac{1}{2} + o\left(1\right)\right) \log n}}. In \cite{luks1982isomorphism}, Luks reduced the bounded valence graph isomorphism problem to the color automorphism problem, and gave a polynomial time algorithm. In another paper \cite{babai1982isomorphism}, Babai et al. created two polynomial algorithms using two different approaches, i.e., the tower of groups method, and the recursion through systems of imprimitivity respectively, for the bounded eigenvalue multiplicity graph isomorphism problem. The isomorphism problem for planar graphs is known to be in polynomial time due to Hopcroft et al. \cite{hopcroft1974linear}. In their paper, the authors used a reduction approach to eventually tranform the graphs into five regular polyhedral graphs and check the isomorphism by exhaustive matching in a fixed finite time. Miller \cite{miller1980isomorphism} used a different approach by finding minimal embeddings of the graphs of bounded genus and checking their isomorphism by generating codes. Babai et al. in \cite{babai1980random} and Czajka et al. in \cite{czajka2008improved} showed that the isomorphism of almost all the graphs in a class of random graphs can be tested in linear time. Both of their approaches exploit the properties of the degree sequence of a random graph. Babai et al. \cite{babai2013faster} proved that while the graph isomorphism problem for strongly regularly graphs may be solved faster than the general version it is still an exponential time algorithm. A series of dramatic events took place recently between $2015$ and $2017$ in the field of graph isomorphism. In December, $2015$ \cite{babai2015graph}, Babai posted a pre-print claiming that the general graph isomorphism problem can be solved in quasipolynomial time. One of the authors of this papers was fortunate enough to witness a live proof session of the algorithm by Babai in Discrete Mathematics $2016$. Two years later, Helfgott \cite{helfgott2017isomorphismes} pointed out a serious error in that proof. Babai immediately fixed the proof and graph isomorphism still remains in quasipolynomial time.

Although there is a quasi-polynomial time algorithm for the general graph isomorphism problem, it is not proven to be optimal. So, the complexity class of the graph isomorphism problem is yet undecided. While it is known that the problem is in {\bf NP} \cite{garey2002computers}, it is not known whether the problem is in {\bf P} or {\bf NP}-complete. This is why the graph isomorphism problem is called an {\bf NP}-intermediate problem. Sch{\"o}ning \cite{schoning1988graph} has shown that graph isomorphism is in \ensuremath{L^P_2} and not \ensuremath{\gamma}-complete under the assumption that the polynomial hierarchy does not collapse to \ensuremath{L^P_2}. Given this information, many researchers believe that the graph isomorphism problem is not {\bf NP}-complete.

While the efforts towards finding an efficient solution for the general graph isomorphism problems have been unsuccessful, the researchers have attempted practically feasible methods to solve the problem in reasonable time frame.

The hidden subgroup approach for both the graph isomorphism and automorphism problems require the computing of the quantum Fourier sampling of the ambient symmetric group. This has been an active area of research since Peter Shor invented the famous Shor's algorithm, a quantum hidden subgroup algorithm for the abelian groups,  to solve prime factorization \cite{shor1999polynomial}. While at this moment, there is no known efficient quantum hidden subgroup algorithm for symmetric groups, researchers have shed some light on why it had been so difficult to find them.

While surveys like \cite{RevModPhys.82.1}, summarizes the advances made so far in the area of hidden subgroup algorithms, it would always be helpful to review the negative results in this section to illustrate why this is a difficult problem. It is noteworthy that all the positive results, so far, have been demonstrated for the synthetically created product groups. While this approach may not have immediate practical application, this idea of creating synthetic groups has been used in this paper to generalize results. One of the first results for the non-abelian hidden subgroup problems was presented by Roetteler et al \cite{roetteler1998polynomial}. In that paper, the authors proved an efficient hidden subgroup algorithm for the wreath product $\mathbb{Z}^k_2 \wr \mathbb{Z}_2$ which is a non-abelian group. Similarly, Ivanyos et al \cite{ivanyos2003efficient} proved the existence of an efficient hidden subgroup algorithm for a more general non-abelian nil-$2$ groups. Later Friedl et al  \cite{friedl2003hidden} generalized the result such that there are efficient hidden subgroup algorithms for the groups whose
derived series have constant length and whose Abelian factor groups are each the
direct product of an Abelian group of bounded exponent and one of polynomial
size. Ettinger et al \cite{ettinger2000quantum} showed that it is possible to reconstruct a subgroup hidden inside the dihedral group using finite number of queries. This result was later generalized by Ettinger et al \cite{ettinger1999hidden} that arbitrary groups may be reconstructed using finite queries but they did not give any specific set of measurement.

In \cite{moore2002hidden}, Moore et al proved that although weak quantum Fourier sampling fails to determine the hidden subgroups of the non-abelian groups of the form $\mathbb{Z}_q \ltimes \mathbb{Z}_p$, where $q \mid \left(p-1\right)$ and $q = p / \text{polylog}\left(p\right)$, strong Fourier sampling is able to do that. Later on, Moore et al  \cite{moore2006generic, moore2005explicit}  proved the existence of $polylog \left(|G|\right)$ sized quantum Fourier circuits for the groups like $S_n$, $H \wr S_n$, where $|H| = \text{poly}\left(n\right)$, and the Clifford groups. The authors also gave the circuits of subexponential size for standard groups like $\text{GL}_n \left(q\right)$, $\text{SL}_n \left(q\right)$, $\text{PGL}_n \left(q\right)$, and $\text{PSL}_n \left(q\right)$, where $q$ is a fixed prime power. Moore et al  \cite{moore2008symmetric} have also presented a stronger result where they have shown that it is not possible to reconstruct a subgroup hidden inside the symmetric group with strong Fourier sampling and both arbitrary POVM and entangled measurement. At the same time, the authors did not rule out the possibility of success using other possible measurements which is still an open question. Bacon et al \cite{bacon2005optimaldi} proved that the so called {\it pretty good measurement} is optimal for the dihedral hidden subgroup problem. Moore et al  \cite{moore2005distinguishing} extended this result for the case where the hidden subgroup is a uniformly random conjugate of a given subgroup. Moore et al \cite{moore2007power} eventually proved a more general results that strong quantum Fourier sampling can reconstruct $q$-hedral groups. Alagic et al \cite{alagic2005strong} proved a general result that strong Fourier sampling fails to distinguish the subgroup of the power of a given non-abelian simple group. Moore et al \cite{moore2005tight} later proved that arbitrary entangled measurement on $\Omega \left(n \log n\right)$ coset states is necessary and sufficient to extract non-negligible information. Similar result was also proved in \cite{hallgren2010limitations} separately. Few years later, Moore et al \cite{moore2010impossibility} proved a negative result that the {\it quantum sieve algorithm}, i.e. highly entangled measurements across $\Omega(n \log n)$ coset states, cannot solve the graph isomorphism problem.

It is important to point out that all the groups used in the previously mentioned results are conveniently chosen and synthetically created. Moreover, they are sporadic so it is not clear how the knowledge can be extrapolated to the symmetric groups. As the graph automorphism problem is Karp-reducible to the graph isomorphism problem, it is believed to be sufficient to investigate the hidden subgroup representation of the graph isomorphism problem. With all these unsuccessful attempts for the last couple of decades presented above, one may ask whether the hidden subgroup approach is the right way to attempt the graph isomorphism problem. If it is, there would have been a Karp-reduction from the hidden subgroup representation of the graph automorphism problem to the hidden subgroup representation of the graph isomorphism problem. This paper gives one such reduction in Section ~\ref{sec:iso-auto}. So, another way of looking at the problem is to understand the hidden subgroup  complexity of the graph automorphism problem and reduce the results to graph isomorphism.

\section{Preliminaries}
\label{sec:prem}
\subsection{Graph Theory}
Most of the work presented in this paper involves the graph isomorphism and automorphism problems. So, it would be appropriate to start the background section with a few concepts of graph theory.  The materials in this section are reproduced from the very well written book by Bollob\'{a}s \cite{bollobas2013modern}. The section does not contain a comprehensive coverage on graph theory, rather they are only related to the discussion of this paper.

\begin{mydef} [Graph]
A graph $\Gamma$ is an ordered pair of disjoint sets $\left(V, E\right)$ such that $E$ is a subset of the set $V^{\left(2\right)}$ of unordered pairs of $V$.
\end{mydef}

$V$ is known as the set of vertices, and $E$ is known as the set of edges. Each element of $E$ connects two elements of $V$. A graph is directed if the edge $\left(v_i, v_j\right)$ is an element of $E$ but $\left(v_j, v_i\right)$ is not for all $i$ and $j$. A simple graph does not have loops or multi-edges. This paper only focuses on questions defined on simple undirected graphs.

\subsubsection{Graph isomorphism and automorphism}
\label{sec:iso-auto}
The graph isomorphism and automorphism problems are the two of the oldest problems in combinatorics. The formal statement of the graph isomorphism problem goes as follows as mentioned in \cite{fortin1996graph}. 

\begin{mydef} [Graph isomorphism ({\bf GI})]
\label{prob:gi}
Given two graphs, \ensuremath{\Gamma_1 = \left(V_1, E_1\right)} and \ensuremath{\Gamma_2 = \left(V_2, E_2\right)},  does there exist a bijection \ensuremath{f : V_1 \to V_2} such that \ensuremath{\forall a, b \in V_1, \left(a, b\right) \in E_1 \iff \left(f\left(a\right), f\left(b\right)\right) \in E_2}?
\end{mydef}
Here, \ensuremath{V_1} and \ensuremath{V_2} are the sets of vertices and \ensuremath{E_1} and \ensuremath{E_2} are the sets of edges of \ensuremath{\Gamma_1} and \ensuremath{\Gamma_2} respectively.

The graph automorphism problem is a special version of Definition \ref{prob:gi} when $\Gamma_1 = \Gamma_2$.  

\begin{mydef} [Graph automorphism ({\bf GA})]
\label{prob:ga}
Given a graph $\Gamma = \left(V, E\right)$, compute the automorphism groups which are $\Gamma \to \Gamma$ isomorphisms; and form the subgroup $\text{Aut}\left(\Gamma\right)$ of the symmetric group $S_{|V|}$.
\end{mydef}

The reducibility from {\bf GA} to {\bf GI} is discussed in the rest of this section based on a few theorems proven in \cite{kobler2012graph}.  The outline is as follows. First, it has been shown that {\bf GA} is Turing-reducible to {\bf GI} i.e. $\text{\bf GA} \le^p_T \text{\bf GI}$ using Algorithm ~\ref{algo:ga-gi-turing-reduction} which is the Example 1.10 of \cite{kobler2012graph}. 

\begin{algorithm}[H]
\caption{$GA \le^p_T GI$}
     \label{algo:ga-gi-turing-reduction}
\begin{algorithmic}[1]
\Procedure {GA-GI-Turing-Reduction}{$\Gamma, n$} \Comment{graph $\Gamma$ with $n$ nodes}

\For{$i\gets 1, n-1 $}
\For{$j\gets i+1, n $}
\If{$\left(\Gamma_{[i]}, \Gamma_{[j]}\right) \in GI$} \Comment{$\Gamma_{[i]}$ denotes a copy of the graph $\Gamma$ with a label attached with node $i$}
\State {\bf accept}
\EndIf
\EndFor
\EndFor
\State {\bf reject};
\EndProcedure
\end{algorithmic}
\end{algorithm}

Then, it has been proven that {\bf GA} is Karp-reducible to {\bf GI}. Instead of reproducing the detailed proof from \cite{kobler2012graph}, this section provides a sketch of it. First, it needs to be shown that {\bf GA} has a polynomial time computable {\it or}-function. Then, it has to be shown that {\bf GI} has both polynomial time computable {\it and-} and {\it or}-functions. Combining these results, it can be argued that {\bf GA} is Karp-reducible to {\bf GI} i.e. $\text{\bf GA} \le^p_m \text{\bf GI}$.

At this point, it is natural to ask whether $\text{\bf GA}_{HSP} \le^p_T \text{\bf GI}_{HSP}$ or ${\bf GA}_{HSP} \le^p_m {\bf GI}_{HSP}$.

It can be trivially shown that the hidden subgroup representation of the graph automorphism problem is Turing-reducible to the graph isomorphism problem by giving two input graphs as the original graph and the candidate automorphism of the original graph. The technique to prove the Karp-reducibility from $\text{\bf GA}_{HSP}$ to $\text{\bf GI}_{HSP}$ was kindly shown to the authors in a public forum by Grochow \cite{Grochow2017}. The sketch of the algorithm is given below.

First, the condition on the the rigidity of the input graphs is relaxed. This makes the case harder.  Instead, the (non-rigid) {\bf GI} as an HSP is described in the same way, but now the goal is to determine the size of the hidden subgroup, or a generating set. The difference between the isomorphic and non-isomorphic cases will be a factor of $2$ in the order of the hidden subgroup. If  the problem is expressed as finding generators of the hidden subgroup, then the question is whether any generator switches $\Gamma_1$ and $\Gamma_2$.

Now, an instance of $\text{\bf GA}_{HSP}$ corresponding to a graph $\Gamma$ is given by the function from $S_n \to M_n(\mathbb{C})$ defined by $f(\pi) = A(\pi(\Gamma))$ where $A(\cdot)$ denotes the adjacency matrix. In particular, $f(e) = A(\Gamma)$. Then the usual Karp reduction is applied from {\bf GA} to {\bf GI} to get a pair of graphs $\Gamma_1, \Gamma_2$. Then an instance of $\text{\bf GI}_{HSP}$ instance, of the type described in the preceding paragraph, can be created corresponding to the pair $\Gamma_1, \Gamma_2$ (that is, the disjoint union $\Gamma_1 \cup \Gamma_2$). Thus ${\bf GA}_{HSP} \le^p_m {\bf GI}_{HSP}$.

\subsection{Representation Theory}
A few concepts of the representation theory are discussed in the current section which are relevant to this paper. The discussion is limited to the representation theory of symmetric groups. A more detailed introduction may be found in \cite{fulton1991representation, curtis1966representation, sagan2013symmetric}. The paper has borrowed the notations and definitions of representation theory generously from the above mentioned standard resources.

\begin{mydef} [Matrix representations \cite{sagan2013symmetric}]
A matrix representation of a group $G$ is a group homomorphism
\begin{align}
X : G \to GL_d.
\end{align}
Equivalently, to each $g \in G$ is assigned $X \left(g\right) \in \text{Mat}_d$ such that
\begin{itemize}
\item $X \left(e\right) = I$ the identity matrix, and
\item $X \left(g h\right) = X \left(g\right) X \left(h\right)$ for all $g, h \in G$.
\end{itemize}
The parameter $d$ is called the degree, or dimension, of the representation and is denoted by $\text{deg } X$.
\end{mydef}

\begin{mydef} [Young diagram]
For any partition $\lambda_1, \ldots, \lambda_k$ of an integer $\lambda$, there is a diagram associated called the Young diagram where there are $\lambda_i$ cells in the $i$-th row. The cells are lined up on the left.
\end{mydef}

\begin{mydef} [Restricted and induced representations]
If $H \subset G$ is a subgroup, any representation $\rho_1$ of $G$ restricts to a representation of $H$, denoted $Res^G_H \rho_1$ or simple $Res \rho_1$. Let $\rho_2 \subset \rho_1$ be a subspace which is $H$-invariant. For any $g$ in $G$, the subspace $g . \rho_2 = \left\{g . w : w \in \rho_2\right\}$ depends only on the left coset of $g H$ of $g$ modulo $H$, since $g h . W = g . \left(h . \rho_2\right) = g. \rho_2$; for a coset $c$ in $G/H$,  $c . \rho_2$ is the subspace of $\rho_1$ subspace of $\rho_1$. $\rho_1$ is induced by $\rho_2$ if every element in $\rho_1$ can be written uniquely as a sum of elements in such translates of $\rho_2$, i.e. 
\begin{align}
\rho_1 &= \bigoplus_{c \in G/H} c . \rho_2
\end{align}
In this case, the induced representation is $\rho_1 = Ind^G_H \rho_2$  = Ind $\rho_2$.
\end{mydef}

A common representation to be seen in later sections of this report is the regular representation \cite{fulton1991representation}.

\begin{mydef} [Regular representation]
If $X$ is any finite set and $G$ acts on the left on $X$, i.e., $G \to Aut \left(X\right)$ is a homomorphism to the permutation group of $X$, there is a associated permutation representation: let $V$ be the vector space with basis $\left\{e_x: x \in X\right\}$, and let $G$ act on $V$ by
\begin{align}
g \cdot \sum a_x e_x = \sum a_x e_{g x}.
\end{align}
The regular representation, denoted $R_G$ or $R$, corresponds to the left action of $G$ on itself.
\end{mydef}

The character of a group element is defined as follows \cite{fulton1991representation}.

\begin{mydef} [Character]
If $\rho$ is a representation of a group $G$, its character $\chi_\rho$ is the complex-valued function on the group defined by 
\begin{align}
\chi_\rho \left(g\right) &= Tr \left(g|_\rho\right),
\end{align}
the trace of $g$ on $\rho$.
\end{mydef}

It is also useful to define the inner product of characters \cite{sagan2013symmetric}.

\begin{mydef}[Inner product of characters]
\label{def:inner-prod-char}
Let $\chi$ and $\psi$ be the characters of a group $G$. The {\it inner product} of $\chi$ and $\Psi$ is
\begin{align}
\langle \chi, \Psi\rangle &= \frac{1}{|G|} \sum_{g\in G} \chi \left(g\right) \Psi^\dagger \left(g\right)
\end{align}
\end{mydef}

The character table of a finite group  is defined as follows \cite{sagan2013symmetric}.

\begin{mydef} [Character table]
Let $G$ be a group. The {\it character table} of $G$ is an array with rows indexed by the inequivalent irreducible characters of $G$  and columns indexed by the conjugacy classes. The table entry in row $\chi$ and column $K$ is $\chi_K$:

\begin{tabular}{ c | c c c }
   & \ldots & $K$ & \ldots \\
  \hline
  \vdots &   & \vdots & \\
$\chi$ & \ldots   & $\chi_K$ & \\
\vdots &    &  & \\
\end{tabular}

By convention, the first row corresponds to the trivial character, and the first column corresponds to the class of the identity, $K = \left\{e\right\}$.
\end{mydef}

Two equivalent procedures are provided for computing the character table of any symmetric group $S_n$ improvising from \cite{Gillespie2012}.

The most straight forward way \cite{sagan2013symmetric} to compute the character table is given in Algorithm ~\ref{algo:char-tab-sagan}.

\begin{algorithm}[H]
\caption{{\bf CHARACTER-TABLE-SAGAN}}
     \label{algo:char-tab-sagan}
\begin{algorithmic}[1]
\Procedure {CHARACTER-TABLE-SAGAN}{$S_n$}
\State Determine all the partitions of $n$ which will also infer the conjugacy classes.

\State Enumerate all group elements and cluster them based on their cycle types. These clusters will coincide with the conjugacy classes.

\State For each class, compute the irreducible representation for each group element.

\State For each class, determine the character of the irreducible representation. All group elements of the same cycle type will have the same character.

\State Populate the table with the characters following the prescribed order of the partitions for both column and rows.
\EndProcedure
\end{algorithmic}
\end{algorithm}

Enumeration of conjugacy classes becomes tedious when  groups larger than $S_5$ are being considered. By using the {\it Murnaghan-Nakayama rule} , The process can be simplified even for larger groups \cite{stanley1986enumerative} as shown by \cite{Gillespie2012} in Algorithm ~\ref{algo:char-tab-mn-rule}.

\begin{algorithm}[H]
\caption{{\bf CHARACTER-TABLE-GILLESPIE}}
     \label{algo:char-tab-mn-rule}
\begin{algorithmic}[1]
\Procedure {CHARACTER-TABLE-GILLESPIE}{$S_n$}
\State   The  conjugacy classes of $S_n$ are the permutations having a fixed number of cycles of each length, corresponding to a partition of $n$ called the shape of the permutation. Since, the characters of a group are constant on its conjugacy classes,  index the columns of the character table by these partitions. The partitions are arranged in an increasing order.

\State There are precisely as many irreducible characters as conjugacy classes, so the irreducible characters can be indexed by the partitions of $n$. Represent each partition as a Young diagram and write them, or the characters directly down, the left of the table in a decreasing order of the partitions.

\Comment{The Murnaghan-Nakayama Rule}
\State  Calculate the entry in row $\lambda$  and column $\mu$. Define a filling of $\lambda$ with content $\mu$ to be a way of writing a number in each square of $\lambda$ such that the numbers are weakly increasing along each row and column and there are exactly $\mu_i$ squares labeled $i$ for each $i$.

\State Consider all fillings of $\lambda$ with content $\mu$ such that for each label $i$, the squares labeled $i$ form a connected skew tableaux that does not contain a $2 \times 2$ square. Such a tableaux is called a {\it border-strip tableaux}.

\State For each label in the tableau, define the height of the corresponding border strip to be one less than the number of rows of the border strip. Weight the tableau by $\left(-1\right)^s$ where $s$ is the sum of the heights of the border strips that compose the tableau.

\State The entry in the character table is simply the sum of these weights.
\EndProcedure
\end{algorithmic}
\end{algorithm}

Now, the concepts of hook is introduced which is used in computing characters of group representations.

\begin{mydef} [Hook]
For a cell $\left(i, j\right)$ of a Young tableau $\lambda$, the $\left(i, j\right)$-hook $h_{i, j}$ is the collection of all cells of $\lambda$ which are beneath $\left(i, j\right)$ (but in the same column) or t the right of $\left(i, j\right)$ (but in the same row), including the cell $\left(i, j\right)$. The {\it length} of the hook is the number of cells appearing in the hook.
\end{mydef}

\begin{mydef} [Skew hook]
A skew hook $s$ of a Young diagram $\lambda$ is a connection collection of boundary boxes such that their removal from $\lambda$ results in a (smaller) diagram.
\end{mydef}

The Murnaghan-Nakayama rule \cite{stanley1986enumerative} is given below.

\begin{mytheo}[The Murnaghan-Nakayama rule]
\label{theo:mnrule}
Let $c$ be a permutation with cycle structure $\left(c_1, \ldots, c_t\right)$, $c_1 \ge \ldots \ge c_t$. Then

\begin{align}
\chi_\lambda \left(c\right) &= \sum_{s_1, \ldots, s_t} \left(-1\right)^{v\left(s_1\right)} \ldots \left(-1\right)^{v\left(s_t\right)},
\end{align}

where each $s_i$ is a skew hook of length $c_i$ of the diagram or partition $\lambda$ after $s_1, \ldots, s_{i - 1}$ have been removed, and $v \left(s_i\right)$ denotes the number of vertical steps in $s_i$.
\end{mytheo}

The semidirect product is defined as follows \cite{dummit2004abstract}.

\begin{mydef} [Semidirect product]
Let $H$ and $K$ be groups, with $K$ acting on $H$ via an action $\phi : K \to Aut \left(H\right)$. The multiplication operation is defined as follows.
\begin{align}
\left(h_1, k_1\right) * \left(h_2, k_2\right) &= \left(h_1 \phi_{k_1} \left(h_2\right), k_1 k_2 \right) =  \left(h_1 \left(k_1 \cdot h_2\right), k_1 k_2\right)
\end{align}.
Then the group $G$ from is called the semidirect product of $H$ by $K$, and is denoted by $H \rtimes_\phi K$.
\end{mydef}

The wreath product is defined as follows \cite{dummit2004abstract}.

\begin{mydef} [Wreath product]
Let $K$ and $L$ be groups, let $n$ be a positive integer, let $\phi : K \to S_n$  be a homomorphism and let $H$ be the direct product of $n$ copies of $L$. Let $\psi$ be an injective homomorphism from $S_n$ into $Auto \left(H\right)$  constructed by letting the elements of $S_n$ permute the $n$ factors of $H$. The composition $\psi \circ \phi$ is a homomorphism from $G$ into $Aut \left(H\right)$. The wreath product of $L$ by $K$ is the semidirect product $H \rtimes K$ with respect to this homomorphism and is denoted by $L \wr K$. 
\end{mydef}

The dimension of an irreducible representation of the symmetric group is defined as follows.

\begin{mydef} [Dimension of an irreducible representation]
The dimension $dim_{\rho_\lambda}$ of an irreducible representation $\rho$ for a partition $\lambda = \left(\lambda_1 + \ldots + \lambda_i + \ldots + \lambda_k\right)$ of a symmetric group $S_n$ is given as follows \cite{fulton1991representation}.

\begin{align}
dim_{\rho_\lambda} &= \frac{n!}{l_1 \cdot \ldots \cdot l_k!} \Pi_{i < j} \left(l_i - l_j\right),
\end{align}
with $l_i = \lambda_i + k - i$.
\end{mydef}

It is suitable to mention the following theorem on the multiplicity of an irreducible representation in the regular representation \cite{fulton1991representation}.

\begin{mytheo}
\label{theo:irrep-multiplicity-in-regular}
Every irreducible representation $\rho$ occurs $\text{dim}\left(\rho\right)$ times in the regular representation.
\end{mytheo}

\begin{proof}[Proof of Theorem ~\ref{theo:irrep-multiplicity-in-regular}]
Let $\chi$ be the character of the regular representation. Then
\begin{align*}
\chi \left(g\right) &=\begin{cases} 
    n       & \quad \text{if } g = 1, \text{ and}\\
   0  & \quad \text{ otherwise.}\\
  \end{cases}
\end{align*}
Because, each group elements acts by a permutation matrix, and the trace of a permutation matrix is simply the number of fixed points of the permutation. Thus,
\begin{align*}
\langle \chi_\rho, \chi \rangle &= \frac{1}{n} \bar{\chi_\rho \left(1\right) \chi \left(1\right)}
\\
&= \frac{1}{n} \text{dim} \left(\rho\right) n
\\
&= \text{dim} \left(\rho\right)
\end{align*}
\qedhere
\end{proof}

Another two important concepts in representation theory are {\it restriction} and {\it induction} \cite{sagan2013symmetric}.

\begin{mydef}[Restriction]
Let $H$ be a subgroup of $G$ and $X$ be a matrix representation of $G$. The restriction of $X$ to $H$, $X \downarrow^G_H$, is given by
\begin{align*}
X\downarrow^G_H \left(h\right) = X \left(h\right)
\end{align*}
for all $h \in H$.
\end{mydef}

\begin{mydef}[Induction]
\label{def:induction}
Let $H \le G$ and $t_1, \ldots, t_l$ be a fixed transversal for the left cosets of $H$, i.e., $G = t_1 H \sqcup \ldots \sqcup t_l H$. If $Y$ is a representation of $H$, then the corresponding induced representation $Y\uparrow^G_H$ assigns to each $g \in G$ the block matrix
\begin{align*}
Y\uparrow^G_H \left(g\right) &= Y \left(t^{-1}_i g t_j\right)
\nonumber\\
&= \begin{pmatrix}
  Y \left(t^{-1}_1 g t_1\right) & Y \left(t^{-1}_1 g t_2\right) & \cdots & Y \left(t^{-1}_1 g t_l\right) \\
  Y \left(t^{-1}_2 g t_1\right) & Y \left(t^{-1}_2 g t_2\right) & \cdots & Y \left(t^{-1}_2 g t_l\right) \\
  \vdots  & \vdots  & \ddots & \vdots  \\
  Y \left(t^{-1}_l g t_1\right) & Y \left(t^{-1}_l g t_2\right) & \cdots & Y \left(t^{-1}_l g t_l\right)
 \end{pmatrix}
\end{align*}
where $Y \left(g\right)$ is the zero matrix if $g \notin H$.
\end{mydef}

It is natural to define the characters for the restricted and induced representations \cite{serre2012linear}.

\begin{mydef}[Character of restricted representation]
Let $X$ be a matrix representation of a group $G$, and let $H \le G$ be a subgroup. Then, the character of the restricted representation $\chi \downarrow^G_H \left(h\right)$ is the character of the original representation $\chi \left(h\right)$ for all $h \in H$.
\end{mydef}

The definition of the character of induced representation is reproduced from \cite{JohnArmstrongCharInducRepre}.

\begin{mydef}[Character of induced representation]
\label{def:char-ind-rep}
Let $Y$ be a matrix representation of a group $H$ such that $H \le G$. A transversal of $H$ in $G$ is now picked. Using the previously mentioned formula for the induced representation, it is found that,
\begin{align*}
\chi \uparrow^G_H \left(g\right) &= Tr \left( Y \left(t^{-1}_i g t_j\right)\right)
\nonumber\\
&= Tr \begin{pmatrix}
  Y \left(t^{-1}_1 g t_1\right) & Y \left(t^{-1}_1 g t_2\right) & \cdots & Y \left(t^{-1}_1 g t_l\right) \\
  Y \left(t^{-1}_2 g t_1\right) & Y \left(t^{-1}_2 g t_2\right) & \cdots & Y \left(t^{-1}_2 g t_l\right) \\
  \vdots  & \vdots  & \ddots & \vdots  \\
  Y \left(t^{-1}_l g t_1\right) & Y \left(t^{-1}_l g t_2\right) & \cdots & Y \left(t^{-1}_l g t_l\right)
 \end{pmatrix}
 \nonumber\\
 &= \sum^n_{i=1} Tr \left(Y \left(t^{-1}_i g t_i\right)\right)
 \nonumber\\
 &= \sum^n_{i=1} \chi  \left(t^{-1}_i g t_i\right)
\end{align*}
where $\chi \left(g\right)$ is the zero matrix if $g \notin H$.
\end{mydef}

Since, $\chi$ is a class function on $H$, conjugation by any element $h \in H$ leaves it the same. So, $\chi \left(h^{-1} g h\right) = \chi \left(g\right)$ for all $g \in G$ and $h \in H$.

The same computation is performed for each element of $H$. Then, all the results are added together and divided by the number of elements of $H$. In other words, the above function is written out in $|H|$ different ways, added  all together, and divided by $|H|$ to get exactly what the section started with started with:

\begin{align}
\chi \uparrow^G_H \left(g\right) &= \frac{1}{|H|} \sum_{h \in H} \sum^n_{i = 1} \chi \left(h^{-1} t^{-1}_i g t_i h\right)
\nonumber\\
&= \frac{1}{|H|} \sum_{h \in H} \sum^n_{i=1} \chi \left(\left(t_i h\right)^{-1} g \left(t_i h\right)\right)
\end{align}
But now as $t_i$ varies over the transversal, and as $h$ varies over $H$, their product $t_i h$ varies exactly once over $G$. That is, every $x \in G$ can be written in exactly one way in the form $t_ h$ for some transversal element $t_i$ and subgroup element $h$. Thus, the following relation can be stated.

\begin{align}
\chi \uparrow^G_H \left(g\right) &= \frac{1}{|H|} \sum_{x \in G} \chi \left(x^{-1} g x\right)
\end{align}

It would be appropriate if the following theorem on the {\it Frobenius reciprocity} is also mentioned in this section \cite{sagan2013symmetric}.

\begin{mytheo}[Frobenius reciprocity]
Let $H \le G$ and suppose that $\psi$ and $\chi$ are characters of     $H$ and $G$, respectively. Then
\begin{align}
\langle \psi \uparrow^G_H, \chi \rangle_G &= \langle \psi, \chi \downarrow^G_H \rangle_H
\end{align}
where the left inner product is calculated in $G$ and the right one in $H$.
\end{mytheo}

A special case of {\it Frobenius reciprocity} is relevant to the discussion of this paper where the representation of $H$ is the trivial representation $\text{\bf 1}_H$ \cite{hallgren2000normal}. The case is described as follows.

\begin{mylem}[Special case of Frobenius reciprocity]
\label{lem:frob-recip-special}
Let $H \le G$ and suppose that $\chi_\rho$ is the character of the irreducible representation $\rho$ of $G$. Then
\begin{align}
\langle {\chi \uparrow^G_H}_{\text{\bf 1}_H}, \chi_\rho \rangle_G &= \langle \chi_{\text{\bf 1}_H}, \chi_\rho \downarrow^G_H \rangle_H
\end{align}
where the left inner product is calculated in $G$ and the right one in $H$.
\end{mylem}

The Example $3.13$ from  \cite{fulton1991representation} can be reproduced here in relevance to the ongoing discussion.

\begin{myrem}
\label{rem:triv-rep-induc}
The permutation representation associated to the left action of $G$ on $G/H$ is induced from the trivial one-dimensional representation $W$ of $H$. Here, the representation of $G$, $V$ has basis $\left\{e_\sigma: \sigma \in G/H\right\}$, and $W = \mathbb{C}\cdot e_H$, with     $H$the trivial coset.
\end{myrem}

Following Remark ~\ref{rem:triv-rep-induc}, it can be said that $\text{\bf 1}_H\uparrow^G_H $ is the permutation representation of $G$. So, according to the Theorem ~\ref{theo:irrep-multiplicity-in-regular}, the multiplicity of $\rho$ in $\text{\bf 1}_H\uparrow^G_H $ is $d_\rho$.

\subsection{Quantum Fourier Sampling}
Quantum Fourier Sampling is a class of quantum algorithms which uses quantum Fourier transformation as a subroutine.
The definition of quantum Fourier transform of a map from a finite group to its representation is defined as follows where the notations are borrowed from \cite{hallgren2000normal}. Later in this section, the algorithm is also presented.

\begin{mydef}[Fourier transformation of a finite group]
\label{def:qft}
Let $f: G \to \mathbb{C}$. The Fourier transform of $f$ at the irreducible representation $\rho$ is the $d_\rho \times d_\rho$ matrix
\begin{align}
\hat{f} \left(\rho\right) &= \sqrt{\frac{d_\rho}{|G|}} \sum_{g\in G} f\left(g\right) \rho \left(g\right)
\end{align}
\end{mydef}

In quantum Fourier transform, the superposition $\sum_{g\in G} f_g |g\rangle$ is identified  with the function $f: G \to \mathbb{C}$ defined by $f \left(f \left(g\right)\right) = f_g$. Using this notation, $\sum_{g\in G} f\left(g\right) |g\rangle$ is mapped under the Fourier transform to $\sum_{\rho \in \hat{G}, 1 \le i, j\le d_\rho} \hat{f}\left(\rho\right)_{i, j}|\rho, i, j\rangle$. Here, $\hat{G}$ is the set of all irreducible representations of $G$ and $\hat{f}\left(\rho\right)_{i,j}$ is a complex number. The probability of measuring the register $|\rho \rangle$ is
\begin{align}
\sum_{1\le i, j\le d_\rho} |\hat{f}\left(\rho\right)_{i,j}|^2 &= \|\hat{f}\left(\rho\right)\|^2
\end{align}
where $\|A\|$ is the natural norm (also known as Frobenius norm) given by $\|A\|^2 = Tr \left(A^\dagger A\right)$.

The Frobenius norm can be calculated from the characters of the group associated which is demonstrated in the following theorem reproduced from \cite{hallgren2000normal}.

\begin{mytheo}
\label{theo:frob-char}
If, $f$ is an indicator function of a left closet of $H$ in $G$, i.e. for some $c \in G$,
\begin{align}
f \left(g\right) &= \begin{cases}
    \frac{1}{\sqrt{|H|}}       & \quad \text{ if } g \in c H, \text{ and }\\
    0  & \quad 0 \text{ otherwise}\\
  \end{cases}
\end{align}, then,

\begin{align}
\|\hat{f} \left(\rho\right)\|^2 &= \frac{|H|}{|G|} d_\rho \langle \chi_\rho, \chi_{\text{\bf 1}_H} \rangle_H
\end{align}
\end{mytheo}

\begin{proof}[Proof of Theorem ~\ref{theo:frob-char}]
From Definition ~\ref{def:qft}, it is known that,
\begin{align}
\|\hat{f} \left(\rho\right)\|^2 &= \sum_{1\le i, j\le d_\rho} |\hat{f}\left(\rho\right)_{i,j}|^2
\end{align}
Only  $\rho$ is to be measured. 

Following relation is assumed in the theorem.
\begin{align}
\|\hat{f} \left(\rho\right)\|^2 &= \| \rho \left(c\right) \sum_{h\in H} \rho \left(h\right) \|^2
\end{align}
$\rho \left(c\right)$ is a unitary matrix. So, as a multiplier it does not change the norm \cite{meyer2000matrix}.
\begin{align}
\|\hat{f} \left(\rho\right)\|^2 &= \| \sum_{h\in H} \rho \left(h\right) \|^2
\end{align}

So, the probability of measuring $\rho$ is determined by $\sum_{h\in H} \rho \left(h\right)$. If correctly normalized, $\frac{1}{|H|}\sum_{h\in H}\rho(h)$ is a projection.

\begin{align}
\left(\frac{1}{|H|}\sum_{h\in H}\rho(h)\right)^2&=\frac{1}{|H|^2}\sum_{h_1,h_2\in H} \rho(h_1 h_2)
\nonumber\\
&=\frac{1}{|H|}\sum_{h\in H}\rho(h)
\end{align}
because $h_1h_2=h$ has $|H|$ solutions $(h_1,h_2)\in H\times H$.

 With the right choice of basis,   $\hat{f} \left(\rho\right)$ will be diagonal and consist of ones and zeros. The probability of measuring $\rho$ will then be the sum of ones in the diagonal. As $\rho$ is an irreducible representation of $G$,  the sum of the matrices $\rho \left(h\right)$ for all $h \in H$ needs to be taken into account. Based of the assumption of the current theorem, one may only consider to evaluate $\rho$ on $H$. According to the assumption, the probability of measuring $\rho$ when $g \notin cH$ is zero. So, one may consider consider $\rho \downarrow^G_H$ instead of $G$.

Then, the Fourier transform of $f$ at $\rho$ is comprised of blocks, each corresponding to a representation in the decomposition of $\rho\downarrow^G_H$. Such as,
\begin{align}
\sum_{h \in H} \rho \left(h\right) &= U  \begin{bmatrix}
\sum_{h \in H} \sigma_1 \left(h\right) & 0 & \cdots & 0 \\
  0 & \sum_{h \in H} \sigma_2 \left(h\right) & \cdots & 0 \\
  \vdots  & \vdots  & \ddots & \vdots  \\
 0 & 0 & \cdots & \sum_{h \in H} \sigma_t \left(h\right)
 \end{bmatrix} U^\dagger
\end{align}
Here, $U$ is an arbitrary unitary transformation, $\sigma_i$ is an irreducible representation of $H$ with possible repetition. Now, as a special case of the orthogonality relation among group characters,  $\sum_{h \in H} \rho \left(h\right)$ is nonzero only when the irreducible representation is trivial, in which case, it is $|H|$.

So, the probability of measuring $\rho$ is:
\begin{align}
\norm{\hat{f} \left(\rho\right)}^2 &= \norm{ \sqrt{\frac{d_\rho}{|G|}} \sum_{h\in H} \rho \left(h\right) }^2
\nonumber\\
&= \frac{d_\rho}{|G|}  \norm{ \sum_{h\in H} \rho \left(h\right) }^2
\nonumber\\
&= \frac{d_\rho}{|G|} \mathrm{tr}\left(\sum_{h_1\in H}\rho(h_1)\right)\left(\sum_{h_2\in H}\rho(h_2)\right)^\dagger
\nonumber\\
&= \frac{d_\rho}{|G|} \mathrm{tr}\sum_{h_1,h_2\in H}\rho(h_1h_2^{-1})
\nonumber\\
&= \frac{d_\rho}{|G|} \frac{1}{|H|} |H|^2 \langle \chi_\rho, \chi_{\text{\bf 1}_H}\rangle_H
\nonumber\\
&= \frac{|H|}{|G|} d_\rho \langle \chi_\rho, \chi_{\text{\bf 1}_H} \rangle_H
\end{align}
It should be mentioned that, by definition, $\rho$ appears  $\langle \chi_\rho, \chi_{1_H} \rangle_H$ times in the decomposition of $\text{\bf 1}_H$.
\qedhere
\end{proof}

Interested readers are encouraged to refer to \cite{diaconis1990efficient} for a review on the classical complexity of Fourier transformation of the symmetric groups. The goal of quantum Fourier sampling algorithm is to sample the labels and elements of the irreducible representations available after quantum Fourier transformation. Sampling only the labels of representations is called weak sampling. On the other hand, sampling also the indices of the elements of the matrix is called strong Fourier sampling. The quantum Fourier sampling algorithm is reproduced from \cite{hallgren2000normal}.

\begin{algorithm}[H]
\caption{{\bf QUANTUM-FOURIER-SAMPLING}}
 \label{algo:weak}
\begin{algorithmic}[1]
\Procedure {QUANTUM-FOURIER-SAMPLING}{$f : G \to S$}
\State Compute $\sum_{g \in G} |g, f \left(g\right) \rangle$ and measure the second register $f \left(g\right)$. The resulting super-position is $\sum_{h \in H} |c h \rangle \otimes | f \left(c h\right)\rangle$ for some coset $c H$ of $H$. Furthermore, $c$ is uniformly distributed over $G$.\;

\State  Compute the Fourier transform of the coset state which is $\sum_{\rho \in \hat{G}} \sqrt{\frac{d_\rho}{|G|}} \sqrt{\frac{1}{|H|}} \left(\sum_{h \in H} \rho \left(c h\right)\right)_{i, j} | \rho, i, j\rangle$, where $\hat{G}$ denotes the set of irreducible representations of $G$.\;

\State  Measure the first register and observe a representation $\rho$ (weak) or $\rho, i, \text{ and } j$ (strong). \;
\EndProcedure
\end{algorithmic}
\end{algorithm}

\section{Quantum Fourier sampling for graph isomorphism}
\label{sec:qft-gi}
The goal of this section is to give the readers a detailed exposition of the standard approach of hidden subgroup algorithms for graph isomorphism problems. In this section, the quantum Fourier sampling algorithm for graph isomorphism is described. Although it is reproduced verbatim from \cite{hallgren2000normal}, the notation is changed to match it to this paper's original discussion of quantum Fourier sampling for graph automorphism. Most of the algorithms for the non-Abelian hidden subgroup problem use a black box for $\varphi$ in the same way as in the Abelian hidden subgroup problem \cite{lomonacopers2002}. This has come to be known as the {\it standard method}. The standard method begins by preparing a uniform superposition over group elements \cite{childs2016lecture}:

\begin{align}
|S_n \rangle := \frac{1}{\sqrt{|S_n|}} \sum_{g\in S_n} |g\rangle
\end{align}

The value of $\varphi \left(g\right)$ is then computed in an ancilla register which creates the following state.

\begin{align}
 \frac{1}{\sqrt{|S_n|}} \sum_{g\in S_n} |g, \varphi \left(g\right) \rangle
\end{align}

Then  the second register is discarded by just being traced it out. If the outcome of the second register is $s$  then the state is projected onto the uniform superposition of those $g \in S_n$ such that $\varphi \left(g\right) = s$. By definition of $\varphi$, it is some left coset of the hidden subgroup $H$. Since every coset contains the same number of elements, each left coset occurs with equal probability. Thus, the standard method produces the following coset state.

\begin{align}
|g H\rangle := \frac{1}{\sqrt{|H|}} \sum_{h \in H} |g h \rangle
\end{align}

or equivalently as the following mixed {\it hidden subgroup state}.

\begin{align}
\rho_H := \frac{1}{|S_n|} \sum_{g \in S_n} |g h \rangle \langle g h |
\end{align}

It has been previously mentioned that $\varphi$ maps the group elements of $S_n$ to $\text{Mat}\left(\mathbb{C}, N \right)$. Here,  more information is presented about the space $\text{Mat}\left(\mathbb{C}, N \right)$. Let the complete set of  irreducible representations of $S_n$ (which are unique up to isomorphism) be $\hat{S_n}$. The Fourier transform is a unitary transformation from the group algebra, $\mathbb{C} S_n$, to a complex vector space whose basis vectors correspond to matrix elements of the irreducible representations of $S_n$, $\oplus_{\rho \in \hat{S_n}} \left(\mathbb{C}^{d_\rho} \otimes \mathbb{C}^{d_\rho}\right)$. Here, $d_\rho$ is the dimension of the irreducible representation $\rho$. 

$|g \rangle$ is the basis vector chosen for the group element $g \in S_n$. There will be $n !$ such basis vectors of dimension $n! \times 1$. For a given group element $g$, there are a particular number of matrices  one for each irreducible representation $\rho$. $|g h\rangle$ is expressed as $|\rho, j, k\rangle$ which is the basis vector labeled by the $\left( j, k \right)$-th element of the irreducible representation $\rho$ of $g$.

As it is mentioned earlier, when only $\rho$ is measured  from $|\rho, j, k\rangle$, it is called {\it weak Fourier sampling}. In {\it strong Fourier sampling}, $j$ and $k$ are also measured.

\subsection{Weak Fourier sampling for $\text{{\bf GI}}_{\text{HSP}}$}
\label{sec:wfs-gi}
This section summarizes what already is known about the weak Fourier sampling when applied to the graph isomorphism problem. The weak Fourier sampling for $\text{{\bf GI}}_{\text{HSP}}$ attempts to measure the labels of irreducible representations of the symmetric group $S_{2 n}$ when the input graphs $\Gamma_1$ and $\Gamma_2$ are of $n$ vertices. It is assumed that $Aut \left(\Gamma_1\right) = Aut \left(\Gamma_2\right) = \left\{e\right\}$. If $\Gamma = \Gamma_1 \sqcup \Gamma_2$, one of the following two claims is true \cite{hallgren2000normal}. 

\begin{itemize}
\item If $\Gamma_1 \not\approx \Gamma_2$, then $\text{Aut}\left(\Gamma \right) = \left\{e\right\}$.
\item If $\Gamma_1 \approx \Gamma_2$, then $\text{Aut}\left(\Gamma \right) = \left\{e, \sigma\right\} = \mathbb{Z}_2$, where $\sigma \in S_{2n}$ is a permutation with $n$ disjoint $2$-cycles.
\end{itemize}

It should be mentioned that in \cite{hallgren2000normal}, the authors derived the success probability of measuring the label of the irreducible representations for $S_n$. In this paper, the same probability will be derived for $S_{2 n}$ to keep consistency with the definition of $\text{{\bf GI}}_{\text{HSP}}$.

The weak Fourier sampling algorithm for finding $\text{Aut}(\Gamma)$ in $S_{2 n}$ is described below.

\begin{algorithm}[H]
\caption{{\bf WEAK-QUANTUM-FOURIER-SAMPLING-}$S_{2 n}$}
 \label{algo:weak-gi}
\begin{algorithmic}[1]
\Procedure {WEAK-QUANTUM-FOURIER-SAMPLING-$S_{2 n}$}{a graph $\Gamma$ such that either $\text{Aut} \left(\Gamma\right) = \left\{e\right\}$ or $\text{Aut} \left(\Gamma\right) = \left\{e, \sigma\right\}$}
\State    Compute $ \frac{1}{\sqrt{\left(2 n\right)!}} \sum_{g\in S_{2 n}} |g, \varphi \left(g\right) \rangle$\;

\State  Compute $\sum_{\rho \in \hat{S_{2 n}}} \sqrt{\frac{d_\rho}{\left(2 n\right)!}} \sqrt{\frac{1}{|\text{Aut}\left(\Gamma \right)|}} \left(\sum_{h \in \text{Aut}\left(\Gamma \right)} \rho \left(c h\right)\right)_{i, j} |\rho, i, j \rangle $\;

\State  Measure $\rho$ as in tracing it out \;
\EndProcedure
\end{algorithmic}
\end{algorithm}

Let $p_\rho$ be the probability of sampling $\rho$ in Algorithm ~\ref{algo:weak-gi} when $\Gamma_1 \not\approx \Gamma_2$, and $q_\rho$ when $\Gamma_1 \approx \Gamma_2$. So,  the induced representation of $\text{Aut} \left(\Gamma\right)$ to $S_{2 n}$, $\text{Ind}^{S_{2 n}}_{\text{Aut} \left(\Gamma\right)} 1$, is the regular representation. So, $\langle \chi_\rho | \chi_{\text{Ind}^{S_{2 n}}_{\text{Aut} \left(\Gamma\right)} 1} \rangle$, the multiplicity of $\rho$ in the regular representation, is $d_\rho$. Hence, $p_\rho = \frac{d^2_\rho}{\left(2 n\right)!}$.

When $\Gamma_1 \approx \Gamma_2$, $\text{Aut} \left(\Gamma\right) = \left\{e, \sigma\right\}$. In this case, the probability of measuring $\rho$, 

\begin{align}
q_\rho &= \frac{|\text{Aut}\left(\Gamma \right)|}{\left(2 n\right)!} d_\rho \langle \chi_1 | \chi_\rho \rangle_{\text{Aut}\left(\Gamma \right)}
\end{align}

$\text{Aut}\left(\Gamma \right)$ has only two elements, $e$ and $\sigma$, hence

\begin{align}
 \langle \chi_1 | \chi_\rho \rangle_{\text{Aut}\left(\Gamma \right)} &= \frac{1}{2} \left(\chi_\rho \left(e\right) + \chi_\rho \left(\sigma\right)\right)
 \nonumber\\
 &= \frac{1}{2} \left(d_\rho + \chi_\rho \left(\sigma\right)\right)
\end{align}

So, 

\begin{align}
q_\rho &= \frac{|\text{Aut}\left(\Gamma \right)|}{\left(2 n\right)!} d_\rho \frac{1}{2} \left(d_\rho + \chi_\rho \left(\sigma\right)\right)
 \nonumber\\
 &= \frac{2}{\left(2 n\right)!} d_\rho \frac{1}{2} \left(d_\rho + \chi_\rho \left(\sigma\right)\right)
 \nonumber\\
 &=\frac{d_\rho}{\left(2 n\right)!}  \left(d_\rho + \chi_\rho \left(\sigma\right)\right)
\end{align}

So, 

\begin{align}
| p_\rho - q_\rho| &=| \frac{d^2_\rho}{\left(2 n\right)!} - \frac{d_\rho}{\left(2 n\right)!}  \left(d_\rho + \chi_\rho \left(\sigma\right)\right)|
 \nonumber\\
 &=  \frac{d_\rho}{\left(2 n\right)!}   |\chi_\rho \left(\sigma\right)|
\end{align}

Now, the Murnaghan-Nakayama rule (Theorem ~\ref{theo:mnrule}) is used to approximate  $|\chi_\rho \left(\sigma\right)|$. 

The number of unordered decompositions for the diagram $\lambda$ with $2 n$ cells is $2^{4 n} = 16^n$. For each unordered decomposition, the number of ordered decomposition is at most $\left(2 \sqrt{2 n}\right)^{\frac{2 n}{2}} = \left(2 \sqrt{2 n}\right)^{n}$ . So, by the Murnaghan-Nakayama rule, $\chi_\rho \left(\sigma\right) \le 16^n \left(2 \sqrt{2 n}\right)^{n}$ .

So,

\begin{align}
|\chi_\rho \left(\sigma\right)| & \le 16^n \left(2 \sqrt{2 n}\right)^{n}
\nonumber\\
 &\le 2^{4 n} 2^{n} \sqrt{2}^n \left( \sqrt{ n}\right)^{n}
 \nonumber\\
 &\le 2^{O\left( n\right)}  n^{\frac{n}{2}}
\end{align}

Now $| p_\rho - q_\rho|$ is computed for all irreducible representations. So,

\begin{align}
| p - q|_1 &= \sum_\rho | p_\rho - q_\rho|
 \nonumber\\
 & \le \sum_\rho \frac{d_\rho}{\left(2 n\right)!} 2^{O\left( n\right)}  n^{\frac{n}{2}}
  \nonumber\\
  & \le \sum_\rho \frac{\sqrt{\left(2n\right)!}}{\left(2 n\right)!} 2^{O\left( n\right)}  n^{\frac{n}{2}}
\nonumber\\
&= \frac{2^{O\left( n\right)}  }{  \frac{n}{2}^{\frac{n}{2}}}
\nonumber\\
&\le \frac{2^{O\left( n\right)}  }{  \left(\frac{n}{2}\right)!}
\nonumber\\
&\lll 2^{-\Omega \left(n\right)}
\end{align}
So, the probability of successfully measuring the labels of the irreducible representations is exponentially low in the size of the graphs.

\section{Quantum Fourier sampling for cycle graph automorphism}
\label{sec:weak-cycle}
In this section, it is shown that quantum Fourier sampling fails to compute the automorphism group of a cycle graph. The result is original to this paper. The scheme of the proof is as follows. First,  the automorphism group of the graph is computed which is trivial for this case. Then  its irreducible representation is computed. Then the general expressions for group characters are derived. And, finally, using the characters, the probabilities of measuring the labels of irreducible representations are computed. A more granular representation of the previously mentioned steps is given below.

            
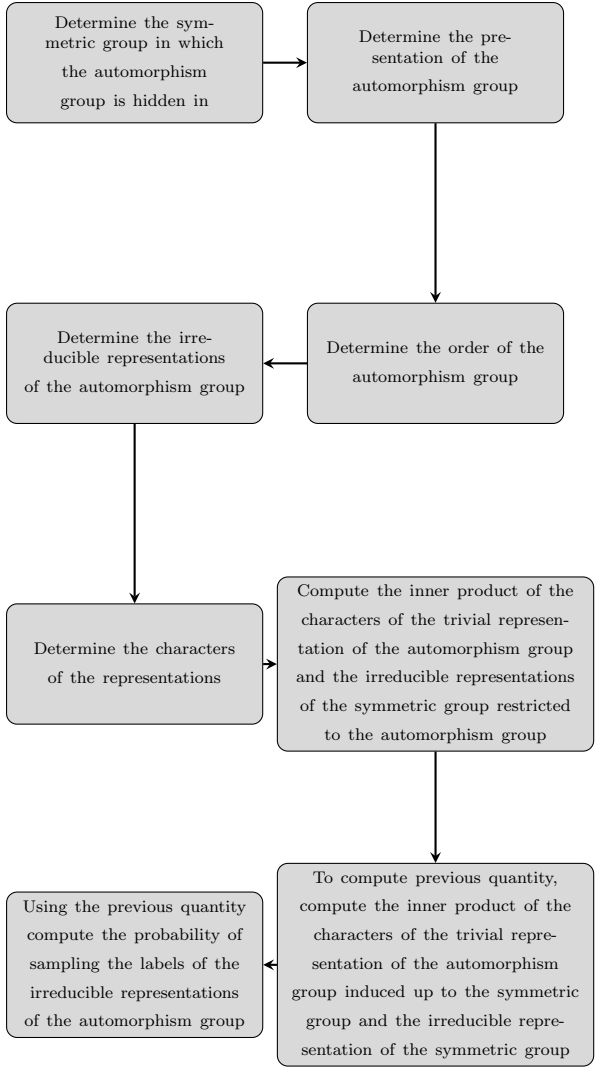
\begin{tikzpicture}[node distance=5cm, scale=0.8, every node/.append style={transform shape}]

\node[text width=4cm] (tf) [startstop] {\scriptsize Determine the symmetric group in which\\ the automorphism group is hidden in};\\

\node[text width=4cm] (ln) [startstop , right of = tf] {\scriptsize Determine the presentation of the\\ automorphism group};

\node[text width=4cm] (mv) [startstop , below of = ln] {\scriptsize Determine the order of the\\ automorphism group};

\node[text width=4cm] (qp) [startstop , left of = mv] {\scriptsize Determine the irreducible representations\\ of the automorphism group};

\node[text width=4cm] (bq) [startstop , below of = qp] {\scriptsize Determine the characters\\ of the representations};

\node[text width=5cm] (bq1) [startstop , right of = bq] {\scriptsize Compute the inner product of the characters of the trivial representation of the automorphism group and the irreducible representations of the symmetric group restricted to the automorphism group};

\node[text width=5cm] (bq2) [startstop , below of = bq1] {\scriptsize To compute previous quantity, compute the inner product of the characters of the trivial representation of the automorphism group induced up to the symmetric group and the irreducible representation of the symmetric group};

\node[text width=4cm] (bq3) [startstop , left of = bq2] {\scriptsize Using the previous quantity compute the probability of sampling the labels of the irreducible representations of the automorphism group};

\draw [arrow] (tf) -- (ln);
\draw [arrow] (ln) -- (mv);
\draw [arrow] (mv) -- (qp);
\draw [arrow] (qp) -- (bq);
\draw [arrow] (bq) -- (bq1);
\draw [arrow] (bq1) -- (bq2);
\draw [arrow] (bq2) -- (bq3);

\end{tikzpicture}

\subsection{Automorphism group of cycle graph}
The automorphism group of an $n$-cycle graphs is the dihedral group $D_n$ of order $2 n$. This section intends to study the weak Fourier sampling of cycle graphs. To provide the background, this section discusses the irreducible representations of the dihedral group $D_n$. 

\begin{mydef}[Cycle Graph]
An $n$-cycle graph is a single cycle with $n$ vertices.
\end{mydef}

The automorphism group of an $n$-cycle graph is the dihedral group $D_n$ which is of order $2 n$.  If $D_n$ is even, the group can be generated as $\langle(2\quad n)(3 \quad n-1) \ldots (\frac{n}{2}-1 \quad \frac{n}{2}+1),\, (1 \ldots n)\rangle$. If $D_n$ is odd, the group can be generated as $\langle(2 \quad n) (3 \quad n-1) \ldots (\frac{n-1}{2} \quad \frac{n+1}{2}),\, (1 \ldots n)\rangle$. This is a manifestation of the presentation 
$D_n = \langle x, y \mid x^n = y^2 = (xy)^2 = 1, y x y=x^{-1}\rangle$. The correspondence consists of $x = (1 \ldots n)$ and  $y = (2 \quad n) (3 \quad n-1) \ldots (\frac{n}{2}-1\quad\frac{n}{2}+1)$ if $n$ is even, and  $y = (2 \quad n) (3 \hspace{0.5cm} n-1) \ldots (\frac{n-1}{2} \quad\frac{n+1}{2})$ if $n$ is odd. The orders of $x$ and $y$ are $n$ and $2$ respectively.

The order of the symmetric group $S_n$ is $n!$. The order of a dihedral group $D_n$ is $2 n$.  So, the index of $D_n$ in $S_n$ is $\frac{n!}{2n} \approx \frac{\sqrt{2 \pi n} \left(\frac{n}
{e}\right)^n}{2n}$.

It would be relevant if the following  important theorem proved in \cite{conrad2009dihedral} is mentioned here.

\begin{mytheo}
\label{theo:d-n-subgroups}
Every subgroup of $D_n$ is cyclic or dihedral. A complete listing of the subgroups is as follows:
\begin{itemize}
\item $\langle x^d \rangle$, where $d \mid n$, with index $2 d$,
\item $\langle x^d, x^i y \rangle$, where $d \mid n$ and $0 \le i \le d-1$, with index $d$.
\end{itemize}
Every subgroup of $D_n$ occurs exactly once in this listing.

In this theorem, subgroups of the first type are cyclic and subgroups of the second type are dihedral: $\langle x^d \rangle \cong \text{\bf Z}/\left(n/d\right)$ and $\langle x^d, x^i y  \rangle \cong D_{n/d}$.
\end{mytheo}

Based on Theorem ~\ref{theo:d-n-subgroups}, following remark can be made.

\begin{myrem}
\label{rem:d-n-subgroup-order}
The order of the subgroups $\langle x^d \rangle$ and $\langle x^d, x^i y \rangle$ are $\frac{n}{d}$ and $\frac{2 n}{d}$ respectively.
\end{myrem}

Every element of $D_n$ is either $x^i$ or $y x^i$ for $0 \le i < n$. The conjugacy classes are small enough in number to be enumerated.

$\begin{array}{rlclcl}
\text{Conjugate} &x^i &\text{by}&x^j &:&(x^j) x^i(x^{-j})=x^i\\
                 &x^i &\text{by}&yx^j&:&(yx^j)x^i(x^{-j}y)= yx^iy=x^{-i}\\
                 &yr^i&\text{by}&r^j &:&(r^j) yr^i(r^{-j}) = yr^{-j}r^ir^{-j} = yr^{i-2j}\\
                 &yx^i&\text{by}&yx^j&:&(yx^j)yx^i(x^{-j}y) = x^{i-2j}y=yx^{2j-i}\\
\end{array}$

The set of rotations decomposes into inverse pairs, $\left\{x^i, \left(x^i\right)^{-1}\right\}$. So, the classes are $\left\{1\right\}, \left\{x, x^{n-1}\right\}$, $\left\{x^2, x^{n-2}\right\}, \ldots$. When $n$ is even, there are $\frac{n}{2}+1$, and when $n$ is odd, there are $\frac{n+1}{2}$ conjugacy classes.

$y x$ is conjugate to $y x^3, y x^5, \ldots$ while $y$ is conjugate to $y x^2, y x^4, \ldots$. If $n$ is even, these two sets are disjoint. However, $y x$ is conjugate to $y x^{n-1}$ (via $x$), so if $n$ is odd, all the non trivial reflections are in one conjugacy class.

So, the total number of conjugacy classes are as follows. If $n$ is even, the total number of conjugacy classes is $\left(\frac{n}{2}+1\right) + 2 = \frac{n}{2}+3$. If $n$ is odd, the total number of conjugacy classes is $\frac{n+1}{2} + 1 = \frac{n+3}{2} $.

The commutators of $D_n$,

\begin{align}
\left[x^i, y x^j\right] &=  x^{-i}\left(y x^j\right)x^i\left(y x^j\right)
\nonumber\\
& = y x^{2i+j}y x^j
\nonumber\\
& = \left(x^i\right)^2
\end{align}

\subsection{Irreducible representations}
\label{sec:irrepdn}
The commutators generate the subgroup of squares of rotations. When $n$ is even, only half the rotations are squares, hence $G/\left[G, G\right]$ is of order four. When $n$ is odd, all rotations are squares, hence $G/\left[G, G\right]$ is of order two. The number of one dimensional irreducible representations is the order of $G/\left[G, G\right]$. So,  when $n$ is even, there are four one dimensional representations and when $n$ is odd, there are two one dimensional representations.

The representations can be enumerated as follows.

\begin{itemize}
\item When $n$ is even:
\begin{itemize}
\item The trivial representation, sending all group elements to the $1 \times 1$ matrix $\begin{pmatrix}1\end{pmatrix}$.
\item The representation, sending all elements in $\langle x \rangle$ to $\begin{pmatrix}1\end{pmatrix}$ and all elements outside $\langle x \rangle$ to $\begin{pmatrix}-1\end{pmatrix}$.
\item The representation, sending all elements in $\langle x^2, y \rangle$ to $\begin{pmatrix}1\end{pmatrix}$ and $x$ to $\begin{pmatrix}-1\end{pmatrix}$.
\item The representation, sending all elements in $\langle x^2, x y \rangle$ to $\begin{pmatrix}1\end{pmatrix}$ and $x$ to $\begin{pmatrix}-1\end{pmatrix}$.
\end{itemize}
\item When $n$ is odd:
\begin{itemize}
\item The trivial representation, sending all group elements to the $1 \times 1$ matrix $\begin{pmatrix}1\end{pmatrix}$.
\item The representation, sending all elements in $\langle x \rangle$ to $\begin{pmatrix}1\end{pmatrix}$ and all elements outside $\langle x \rangle$ to $\begin{pmatrix}-1\end{pmatrix}$.
\end{itemize}
\end{itemize}

The two dimensional irreducible representations are described as follows. There is an obvious subgroup $\{1, x, \ldots, x^{n-1}\}$ which is a cyclic group of order $n$. It can be defined as $C_n < D_{n}$. Since $C_n$ is abelian, it has $n$ irreducible 1-dimensional representations over $\mathbb{C}$, namely

\begin{align}
x\mapsto e^{2\pi ki/n},\qquad 0 &\leq k < n
\end{align}

which captures the idea of rotating by an angle of $2\pi k/n$. These easily-described representations are induced to $D_{n}$ in order to find some possibly new representations.

For the representation $W$ of a subgroup $H\leq G$ (i.e. an H-linear action on $W$), the induced representation of $W$ is 

\begin{align}
\bigoplus_{g\in G/H}g\cdot W
\end{align}

where $g$ ranges over a set of representatives of $G/H$.

The induced representation of $C_n$ to $D_{n}$ for fixed $k$ is straight forward since $D_{n}/C_n$ has representatives $\{1, y\}$. So, one just need to describe the $D_{n}$-vector space $\mathbb{C}\oplus y \cdot\mathbb{C}$ where $\mathbb{C}$ has basis consisting only of $w_1$. Now, the $D_{n}$ action turns into an actual matrix representation.

Specifically, it can be found out how $x$ acts on each summand using the representation of $C_n$: 

\begin{align}
x\cdot w_1 &= e^{2\pi ki/n}w_1, \text{ and}
\end{align}

\begin{align}
x \cdot(y \cdot w_1) &= x y \cdot w_1 
\nonumber\\
&= y x^{-1}\cdot w_1 = e^{-2\pi ki/n} y \cdot w_1
\end{align}

 which means $x$ acts by the matrix 
 $\begin{pmatrix}
 e^{2\pi ki/n}&0\\
 0&e^{-2\pi ki/n}
 \end{pmatrix}
  $.

It can also be figured out how $y$ acts. $y$ obviously takes $w_1$ to $y \cdot w_1$, and $y$ takes $y \cdot w_1$ to $y^2 w_1=w_1$, so $y$ simply interchanges the two summands. This entails  that $y$ acts by the matrix 
$\begin{pmatrix}
0&1\\
1&0
\end{pmatrix}
$.

Here,  the $k$-th two dimensional irreducible representations are listed for the general group elements.

\begin{align}
x \mapsto \begin{pmatrix}
e^{\frac{2 \pi i k}{n}}&0\\
0&e^{-\frac{2 \pi i k}{n}}
\end{pmatrix}
\nonumber\\
x^l \mapsto \begin{pmatrix}
e^{\frac{2 \pi i k l}{n}}&0\\
0&e^{-\frac{2 \pi i k l}{n}}
\end{pmatrix}
\nonumber\\
y \mapsto \begin{pmatrix}
0&1\\
1&0
\end{pmatrix}
\nonumber\\
x^l y \mapsto \begin{pmatrix}
0&e^{\frac{2 \pi i k l}{n}}\\
e^{-\frac{2 \pi i k l}{n}}&0
\end{pmatrix}
\end{align}

It is observed that, $0 \le l \le n-1$. Both $l = 0$ and $l = n$ determine the identity matrix to which the identity element, $e$, is mapped. When $n$ is even, the $k$-th and $\left(n-k\right)$-th representations are equivalent, hence the distinct representations are found only for $k = 1, 2, \ldots, \frac{n-2}{2}$. The representations for $k = 0$ and $k = \frac{n}{2}$ are not irreducible and they decompose into one dimensional representations. On the other hand, when, $n$ is odd, there are $\frac{n-1}{2}$ irreducible representations.

Using the previous calculations,  the total number of irreducible representations for $D_n$ can be computed. When $n$ is even, the total number is $\frac{n-2}{2} + 4 = \frac{n}{2} + 3$. When $n$ is odd, it is $\frac{n-1}{2} + 2 = \frac{n+3}{2}$. 

At this point, the following remark regarding the characters of the irreducible two dimensional representations can be made. 

\begin{myrem}
\label{rem:char-dn}
The characters of the representations of the elements of type $y$ and $x^l y$ are both zeros. The representations of the elements of type $x$ have the same character, $2 \cos \left(\frac{2 \pi k}{n}\right)$. Finally, the representations of the elements of type $x^l$ have the same character, $2 \cos \left(\frac{2 \pi k l}{n}\right)$.
\end{myrem}

\subsection{Sampling the representations}
The interest of this paper with the dihedral group $D_n$, of order $2 n$, is based on the fact that it is the automorphism group of the $n$-cycle graph. A $p$-group is a group where the order of every group element is a power of the prime $p$.  So, $D_n$ can be a $p$-group only when $n = 2^m$, when $m \in \mathbb{Z}$, because that is when $x^{2^m} = y^{2^1} = 1$ for $p = 2$. 

The restrictions from the irreducible representations of $S_n$ to $D_n$ is straight forward. For any irreducible representation $\rho$ of $S_n$, its restriction to $D_n$ is $\rho \left(h\right)$ for all $h \in D_n$. This new representation may not be necessarily irreducible.

The inductions from the irreducible representations of $D_n$ to $S_n$ can be computed following the Definition ~\ref{def:induction}. This new representation may also not be necessarily irreducible. 

Following \cite{hallgren2000normal}, this paper seeks to compute the induced representation of  $\text{\bf 1}_{D_n}$ which is the trivial representation of the dihedral group $D_n$. As a prerequisite,  a transversal for the left cosets of $D_n$ needs to be computed which can be done by any of the two classical algorithms presented in Section $4.6.7$ of \cite{holt2005handbook}. The computation of  a transversal of a group requires the computation of the base of a group which can be computed in polynomial time using the Schreier-Sims algorithm \cite{seress2003permutation, sims1970computational, knuth1991efficient}.

As it has been mentioned previously in the current section, there are $l = \frac{n!}{2n} \approx \frac{\sqrt{2 \pi n} \left(\frac{n}{e}\right)^n}{2n}$ cosets for $D_n$ in $S_n$. This will also be the number of elements in a transversal for the left cosets of $D_n$ in $S_n$. Let the transversal be $t_1, \ldots, t_l$. So, $S_n = t_1 D_n \sqcup \ldots \sqcup t_l D_n$.  $\text{\bf 1}_{D_n} \uparrow^{S_n}_{D_n}$ can be computed following Definition ~\ref{def:induction}. 

It would be instructive to discuss the character table of $S_n$ and $D_n$ here.

Computing the character table of $S_n$ starts from computing the partitions $\lambda_1 \ge \lambda_2 \ge \ldots \ge \lambda_r$ given $\sum_i \lambda = n$. These partitions can be partially ordered as follows. If there are two partitions $\lambda = \left(\lambda_1 \ge \lambda_2 \ge \ldots\right)$ and $\mu = \left(\mu_1 \ge \mu_2 \ge \ldots\right)$, $\lambda \ge \mu$ if,
\begin{align}
\lambda_1 &\ge \mu_1
\nonumber\\
\lambda_1 + \lambda_2 &\ge \mu_1 + \mu_2
\nonumber\\
\lambda_1 + \lambda_2 + \lambda_3 &\ge \mu_1 + \mu_2 + \mu_3
\nonumber\\
&\vdots
\end{align}

The columns of the character table are indexed by the conjugacy classes such that the partitions are arranged in increasing order. On the other hand, the rows are indexed by the characters in the decreasing order of the partitions. Each cell in the table then contains the corresponding character.

The Lemma ~\ref{lem:frob-recip-special} can be applied to obtain the Frobenius reciprocity for $D_n < S_n$.

\begin{align}
\langle \chi_{\text{\bf 1}_{D_n}}, \chi_\rho \downarrow^{S_n}_{D_n} \rangle_{D_n} &= \langle {\chi \uparrow^{S_n}_{D_n}}_{\text{\bf 1}_{D_n}} , \chi_\rho \rangle_{S_n}
\end{align}
where the left inner product is calculated in $S_n$ and the right one in $D_n$. So, if  $\langle \chi_{\text{\bf 1}_{D_n}}, \chi_\rho \downarrow^{S_n}_{D_n} \rangle_{D_n}$ needs to be determined, it would be sufficient to determine $\langle {\chi \uparrow^{S_n}_{D_n}}_{\text{\bf 1}_{D_n}} , \chi_\rho \rangle_{S_n}$.

Following the Definition ~\ref{def:inner-prod-char}, 

\begin{align}
\langle {\chi \uparrow^{S_n}_{D_n}}_{\mathbf{ 1}_{D_n}} , \chi_\rho \rangle_{S_n} &= \sum_{g_i \in S_n}  {\chi \uparrow^{S_n}_{D_n}}_{\mathbf{ 1}_{D_n}} (g_i) \chi^\dagger_\rho (g_i)  
\nonumber\\
&=\sum_{g_i \in S_n} \left(\sum^n_{i=1} \delta_{{{ \chi \uparrow^{S_n}_{D_n}}_{\mathbf{ 1}_{D_n}} (g_i)}_{i,i}}\right)   \chi^\dagger_\rho (g_i)  
\end{align}

The Theorem ~\ref{theo:frob-char} is applied to determine the probabilities of measuring the irreducible representations of $D_n$ through proving a few lemmas. 

\begin{mylem}
\label{lem:d-n-1-d-rep-prob-zero}
The probability of measuring the labels of one dimensional irreducible representations of $D_n$ is zero for non-trivial representations.
\end{mylem}

\begin{proof}[Proof of Lemma ~\ref{lem:d-n-1-d-rep-prob-zero}]
First, it is assumed that $n$ is even. So, there are four one dimensional representations and $\frac{n-2}{2}$ two dimensional representations. The probability of measuring the one dimensional representations is then computed. Let the representations be denoted as $\rho_1$, $\rho_2$, $\rho_3$, and $\rho_4$. So, their dimensions are all the same i.e. $d_{\rho_1} = d_{\rho_2} = d_{\rho_3} = d_{\rho_4} = 1$. Let the probability of measuring $\rho_i$ be $p_{\rho_i}$. So, following Theorem ~\ref{theo:frob-char},

\begin{align}
p_{\rho_i} &= \frac{|D_n|}{|S_n|} d_{\rho_i} \langle \chi_{\rho_i}, \chi_{\text{\bf 1}_{D_n}} \rangle_{D_n}
\nonumber\\
&= \frac{2 n}{n!}  \langle \chi_{\rho_i}, \chi_{\text{\bf 1}_{D_n}} \rangle_{D_n}
\end{align}

$D_n$ has $2 n$ elements. So, 

\begin{align}
\langle \chi_{\rho_i}, \chi_{\text{\bf 1}_{D_n}} \rangle_{D_n} &= \frac{1}{|D_n|} \sum_{g\in D_n} \chi_{\rho_i} \left(g\right) \chi^\dagger_{\text{\bf 1}_{D_n}} \left(g\right)
\nonumber\\
&= \frac{1}{2 n} \sum_{g\in D_n} \chi_{\rho_i} \left(g\right) 
\end{align}

When $i = 1$, $\rho_1$ is the trivial representation which sends all group elements to the $1 \times 1$ matrix $\begin{pmatrix}1\end{pmatrix}$. The probability of measuring this representation is given below.

\begin{align}
\langle \chi_{\rho_1}, \chi_{\text{\bf 1}_{D_n}} \rangle_{D_n} &= \frac{1}{2 n} \sum_{g\in D_n} \chi_{\rho_1} \left(g\right) 
\nonumber\\
&= 1
\end{align}

So, 
\begin{align}
p_{\rho_1} &=  \frac{2 n}{n!}  \langle \chi_{\rho_1}, \chi_{\text{\bf 1}_{D_n}} \rangle_{D_n}
\nonumber\\
&=  \frac{2 n}{n!}  
\end{align}

When $i = 2$, $\rho_2$ is the representation which sends all elements in $\langle x \rangle$ to $\begin{pmatrix}1\end{pmatrix}$ and all elements outside $\langle x \rangle$ to $\begin{pmatrix}-1\end{pmatrix}$. The probability of measuring this representation is given below.

\begin{align}
\langle \chi_{\rho_2}, \chi_{\text{\bf 1}_{D_n}} \rangle_{D_n} &= \frac{1}{2 n} \sum_{g\in D_n} \chi_{\rho_2} \left(g\right) 
\end{align}

As observed from the presentation of $D_n$, the number of elements in $\langle x \rangle$ is $n$. So, $n$ elements will be mapped to $1$ and $n$ elements will e mapped to $-1$. So,

\begin{align}
\langle \chi_{\rho_2}, \chi_{\text{\bf 1}_{D_n}} \rangle_{D_n} &= \frac{1}{2 n} \left(n \left(1\right) + n \left(-1\right)\right) 
\nonumber\\
&= 0
\end{align}

So, 
\begin{align}
p_{\rho_2} &=  \frac{2 n}{n!}  \langle \chi_{\rho_2}, \chi_{\text{\bf 1}_{D_n}} \rangle_{D_n}
\nonumber\\
&= 0
\end{align}

So, the weak Fourier sampling will not be able to determine the labels of the sign representation which sends all elements in $\langle x \rangle$ to $\begin{pmatrix}1\end{pmatrix}$ and all elements outside $\langle x \rangle$ to $\begin{pmatrix}-1\end{pmatrix}$.

When $i = 3$, $\rho_3$ is the representation which sends all elements in $\langle x^2, y \rangle$ to $\begin{pmatrix}1\end{pmatrix}$ and $x$ to $\begin{pmatrix}-1\end{pmatrix}$. The probability of measuring this representation is given below.

\begin{align}
\langle \chi_{\rho_3}, \chi_{\text{\bf 1}_{D_n}} \rangle_{D_n} &= \frac{1}{2 n} \sum_{g\in D_n} \chi_{\rho_3} \left(g\right) 
\end{align}

Following Remark ~\ref{rem:d-n-subgroup-order}, $\rho_3$ sends $n$ elements of $D_n$ to $\begin{pmatrix}1\end{pmatrix}$ and $2 n - n = n$ elements of $D_n$ to $\begin{pmatrix}-1\end{pmatrix}$. So,

\begin{align}
\langle \chi_{\rho_3}, \chi_{\text{\bf 1}_{D_n}} \rangle_{D_n} &= 0
\end{align}

So, the weak Fourier sampling will not be able to determine the labels of the sign representation which sends all elements in $\langle x^2, y \rangle$ to $\begin{pmatrix}1\end{pmatrix}$ and all elements outside $\langle x \rangle$ to $\begin{pmatrix}-1\end{pmatrix}$.

When $i = 4$, $\rho_4$ is the representation which sends all elements in $\langle x^2, x y \rangle$ to $\begin{pmatrix}1\end{pmatrix}$ and $x$ to $\begin{pmatrix}-1\end{pmatrix}$. The probability of measuring this representation is given below.

\begin{align}
\langle \chi_{\rho_4}, \chi_{\text{\bf 1}_{D_n}} \rangle_{D_n} &= \frac{1}{2 n} \sum_{g\in D_n} \chi_{\rho_4} \left(g\right) 
\end{align}

Following Remark ~\ref{rem:d-n-subgroup-order}, $\rho_4$ sends $n$ elements of $D_n$ to $\begin{pmatrix}1\end{pmatrix}$ and $2 n - n = n$ elements of $D_n$ to $\begin{pmatrix}-1\end{pmatrix}$. So,

\begin{align}
\langle \chi_{\rho_4}, \chi_{\text{\bf 1}_{D_n}} \rangle_{D_n} &= 0
\end{align}

So, the weak Fourier sampling will not be able to determine the labels of the sign representation which sends all elements in $\langle x^2, x y \rangle$ to $\begin{pmatrix}1\end{pmatrix}$ and all elements outside $\langle x \rangle$ to $\begin{pmatrix}-1\end{pmatrix}$.

The case when $n$ is odd may be considered as a special case of when $n$ is even and show that only the trivial representation can be sampled with non-zero probability.
\end{proof}

Now, the probability of measuring the labels of two dimensional irreducible representations will be computed in weak Fourier sampling. The discussion starts with the following lemma.

\begin{mylem}
\label{lem:d-n-2-d-rep-prob-zero}
The probability of measuring the labels of two dimensional irreducible representations of $D_n$ is always zero.
\end{mylem}

\begin{proof}[Proof of Lemma ~\ref{lem:d-n-2-d-rep-prob-zero}]
When $n$ is even, there are $\frac{n-2}{2}$ such irreducible representations. Let the irreducible representations be denoted as $\sigma_1, \sigma_2, \ldots, \sigma_k, \dots, \sigma_{\frac{n-2}{2}}$.

The probability of measuring $\sigma_k$ is given below.

\begin{align}
\langle \chi_{\sigma_k}, \chi_{\text{\bf 1}_{D_n}} \rangle_{D_n} &= \frac{1}{2 n} \sum_{g\in D_n} \chi_{\sigma_k} \left(g\right) 
\end{align}

Following Remark ~\ref{rem:char-dn}, $\sigma_k$ maps $n$ number of elements to the matrices for which the characters of the representations are zero. For the rest $n$ number  of the group elements, the character is $2 \cos \left(\frac{2 \pi k l}{n}\right)$ where $0\le l \le n-1$. So,

\begin{align}
\langle \chi_{\sigma_k}, \chi_{\text{\bf 1}_{D_n}} \rangle_{D_n} &= \frac{1}{2 n} \sum^{n-1}_{l = 0} 2 \cos \left(\frac{2 \pi k l}{n}\right)
\nonumber\\
&= \frac{1}{ n} \sum^{n-1}_{l = 0}  \cos \left(\frac{2 \pi k l}{n}\right)
\end{align}

The following formula for the sum of series of cosines when they are in arithmetic progression is worth mentioning as they have been proven in both \cite{knapp2009sines} and ~\cite{holdener2009math}.

\begin{align}
\sum_{l=1}^{n} \cos (l\theta)=\frac{\sin(n\theta/2)}{\sin(\theta/2)}\cos ((n+1)\theta/2),\quad \sin(\theta/2)\neq0.
\end{align}

When $\theta=\dfrac{2 \pi k}{n}$,

\begin{align}
\sum^n_{l=1} \cos \left(\frac{2 \pi k l}{n}\right) &=\frac{\sin(\pi k)}{\sin(\pi k/n)}\cos ((n+1)\pi k/n)= 0
\end{align}

If the interval of $l$ is changed from $\left[1, n\right]$ to $\left[0, n-1\right]$ the sum still remains zero.

\begin{align}
\sum^{n-1}_{l=0} \cos \left(\frac{2 \pi k l}{n}\right) &= \cos (0)+0-\cos(2\pi k)=0.
\end{align}

So, the probability of measuring the labels of the two dimensional irreducible representations is:

\begin{align}
\langle \chi_{\sigma_k}, \chi_{\text{\bf 1}_{D_n}} \rangle_{D_n} &= \frac{1}{ n} \sum^{n-1}_{l = 0}  \cos \left(\frac{2 \pi k l}{n}\right)
\nonumber\\
&= 0
\end{align}

The case of $n$ being odd can be considered as a special case of $n$ being even and the same result can be proved.
\end{proof}

So, the weak Fourier sampling algorithm cannot determine the labels of any of the two dimensional irreducible representations. To summarize, weak Fourier sampling cannot determine any irreducible representation other than the trivial one. So, it cannot determine the automorphism group of the cycle graph which is a trivial problem in the classical paradigm. Lemmas ~\ref{lem:d-n-1-d-rep-prob-zero}  and ~\ref{lem:d-n-2-d-rep-prob-zero} may be consolidated into the following theorem.

\begin{mytheo}
\label{theo:cycle-graph-auto-weak-fails}
Weak quantum Fourier sampling fails to solve the cycle graph automorphism problem.
\end{mytheo}

\begin{proof}
Follows directly from the proofs of the Lemmas ~\ref{lem:d-n-1-d-rep-prob-zero}  and ~\ref{lem:d-n-2-d-rep-prob-zero}.
\end{proof}

It is observed that the success probability of strong Fourier sampling is a conditional probability which depends on the success probability of measuring the labels of representation. As the success probability of measuring the non-trivial representations of $D_n$ is zero, the success probability of measuring their individual matrix elements is also zero.

\begin{mycor}
\label{cor:cycle-graph-auto-strong-fails}
Strong quantum Fourier sampling fails to solve the cycle graph automorphism problem.
\end{mycor}

So far, the discussion in this paper has been limited to the POVMs in the computational basis. One could ask if the results are comparable if one  were performing entangled measurements introduced in \cite{hallgren2010limitations, moore2005tight}. It can be argued that quantum Fourier transform is guaranteed to fail also in the latter type of measurement. To perform an entangled measurement, it is assumed that one has already measured the labels of the irreducible representations with nonzero probability which is not possible for cycle graphs. Hence, the automorphism group of a cycle graph cannot be computed using the measurements in either computational or entangled bases. It is also natural to ask how the results relate to the results of \cite{radhakrishnan2005power} on using random bases. As the bases are always being chosen from the complete set of computational bases, one can argue that strong random Fourier sampling does not change the probability of measuring the labels of irreducible representations from zero (where applicable) to a larger value. It can also be argued in a similar manner for both Kuperberg sieve \cite{kuperberg2005subexponential, moore2010impossibility} and Pretty Good Measurement (PGM) as both algorithms are conditioned on weak sampling first. The Ettinger-H{\o}yer-Knill theorem \cite{ettinger2004quantum} that the quantum query complexity of the hidden subgroup problem is polynomial may also be mentioned here. The theorem cannot be applied here either as the assumption of the theorem is that the quantum algorithm will always output a subset of the hidden subgroup.

\section{Quantum Fourier transform for other graph automorphism problems}
\label{sec:qft-ga-plus}
In Section ~\ref{sec:weak-cycle}, it has been shown that the  Fourier sampling fails to determine the automorphism group of the cycle graphs. It is natural to ask whether same is the case for the graphs which have the dihedral group as a subgroup in it's automorphism group. The question can be answered in the affirmative. Those cases can be instantiated using any of the following three approaches.

\subsection{Inductive approach}
This approach starts with a graph whose automorphism group is a dihedral group. Then, more complicated structures are inductively built on that graph. This is based of the following theorem reproduced from \cite{dresselhaus2007group} and also mentioned as Theorem $10$ in Section $3.2$ in \cite{serre2012linear}. To initiate the discussion,  an example can be provided, where a more complex graph is constructed from cycle graphs and show that quantum Fourier sampling is still guaranteed to fail to compute the automorphism group.

\begin{myexamp}[QFT for the automorphism group of $C_m \sqcup C_n$]
\label{examp:c-m-c-n}
The graph $C_m \sqcup C_n$ can be visualized as follows.

It is already known from \cite{ganesan2012automorphism, jordan1869assemblages} that the automorphism group of $C_m \sqcup C_n$ is $D_m \times D_n$ where the dihedral groups are of order $2 m$ and $2 n$ respectively.

Following the schemes introduced in Section ~\ref{sec:weak-cycle},  the probabilities of measuring the labels irreducible representations of $D_m \times D_n$ are computed.

The Theorem $10 (ii)$ in \cite{serre2012linear} indicates that every irreducible representation of $D_m \times D_n$ can be determined from the irreducible representations of $D_m$ and $D_n$.

The general expression for irreducible representations of dihedral groups is used as it is  given in Section ~\ref{sec:weak-cycle}. Let the $i$-th irreducible $j$-dimensional representation of a dihedral group $D_n$ be $\rho_{i,j, n}$ where $n$ is even. For example, for $D_m$, the irreducible representations are $\rho_{1,1, m}$, $\rho_{2,1, m}$, $\rho_{3,1, m}$, $\rho_{4,1, m}$, $\rho_{1,2, m}$, $\rho_{2,2, m}$, $\rho_{3,2, m}$, and $\rho_{4,2, m}$. The expressions of these representations can be derived from Section ~\ref{sec:irrepdn}. It needs to be noted that in the two dimensional irreducible representations, there is an additional ordering parameter $k$. Hence, $k_m$ and $k_n$ will be used for the groups $D_m$ and $D_n$ respectively. So, for example, $\begin{pmatrix}
e^{\frac{2 \pi i k_{n}}{n}}&0\\
0&e^{-\frac{2 \pi i k_{n}}{n}}
\end{pmatrix}$ denotes the $k_{n}$-th two dimensional irreducible representation for the element $x$. Another parameter $l$ is used to identify the elements $x^l$ in the group $D_n$ where $0 \le l \le n-1$. Let the parameter be $l_m$ for $D_m$ and $l_n$ for $D_n$.

Now, the irreducible representations of $D_m \times D_n$ are enumerated. The corresponding characters are also computed at the same time. There are sixteen one dimensional, thirty two two dimensional, and sixteen four dimensional irreducible representations. Table ~\ref{table:irrep-cn-cm} summarizes the characters of the irreducible representations of $D_m \times D_n$.  Curious readers may refer to Appendix ~\ref{app:cmcn} for detailed derivations.

\begin{table}[H]
\centering
\begin{tabular}{ c| c   } 
 \hline
Irreducible representation  & Character\\ 
 \hline\hline
  $\rho_{7, 8, 15, 16, 23, 24, 31, 32, 39, 40, 47-64} $ &  $\chi_{7, 8, 15, 16, 23, 24, 31, 32, 39, 40, 47-64} = 0$\\
   $\rho_{1}$ & $\chi_{1} = 1$ \\
 $\rho_{2-4, 9-12, 17-20, 25-28}$ & $\chi_{2-4, 9-12, 17-20, 25-28} = \pm 1$ \\
 $\rho_5 $ &  $\chi_5 =  2 \cos \left(\frac{2 \pi k_m}{m}\right)$\\
 $\rho_{6, 22} $ &  $\chi_{6, 22} = 2 \cos \left(\frac{2 \pi k_m l_m}{m}\right)$\\
 $\rho_{13, 21, 29} $ & $\chi_{13, 21, 29} = \pm 2 \cos \left(\frac{2 \pi k_m}{m}\right)$\\
 $\rho_{14, 30} $ & $\chi_{14, 30} = \pm  2 \cos \left(\frac{2 \pi k_m l_m}{m}\right)$\\
 $\rho_{33} $ & $\chi_{33} = 2 \cos \left(\frac{2 \pi k_n}{n}\right)$\\
 $\rho_{34-36}$ & $\chi_{34-36} =\pm 2 \cos \left(\frac{2 \pi k_n}{n}\right)$\\
 $\rho_{37} $ & $\chi_{37} = 4 \cos \left(\frac{2 \pi  k_m}{m}\right) \cos \left(\frac{2 \pi  k_n}{n}\right)$\\
 $\rho_{38} $ & $\chi_{38} = 4 \cos \left(\frac{2 \pi  k_n}{n}\right) \cos \left(\frac{2 \pi  k_m l_m}{m}\right)$\\
 $\rho_{41} $ & $\chi_{41} = 2 \cos \left(\frac{2 \pi k_n l_n}{n}\right)$\\
 $\rho_{42-44} $ & $\chi_{42-44} =\pm 2 \cos \left(\frac{2 \pi k_n l_n}{n}\right)$\\
 $\rho_{45}$ & $\chi_{45} = 4 \cos \left(\frac{2 \pi  k_m}{m}\right) \cos \left(\frac{2 \pi  k_n l_n }{n}\right)$\\
 $\rho_{46} $ & $\chi_{46} = 4 \cos \left(\frac{2 \pi  k_m l_m}{m}\right) \cos \left(\frac{2 \pi  k_n l_n }{n}\right)$\\
\hline
\end{tabular}
\caption{Irreducible representations of $\text{Aut} \left(C_m \sqcup C_n\right)$}
\label{table:irrep-cn-cm}
\end{table}

It can be readily seen that the quantum Fourier transform is guaranteed to fail to measure the labels of   $\rho_{7, 8, 15, 16, 23, 24, 31, 32, 39, 40, 47-64}$. With little more algebraic steps, it can be shown to be true for few more irreducible representations.
\end{myexamp}

Following Example ~\ref{examp:c-m-c-n}, one can create arbitrarily large classes of graph automorphism problem by based on the known Theorem ~\ref{lab:irrep-prod-grp} (not a result original to this paper) for which quantum Fourier transform will always fail. The proof of the Theorem ~\ref{lab:irrep-prod-grp} is given as Theorem $10$ in \cite{serre2012linear}

\begin{mytheo}
\label{lab:irrep-prod-grp}
The direct product of two irreducible representations of groups $H$ and $K$ yields an irreducible representation of the direct product group so that all irreducible representations of the direct product group can be generated from the irreducible representations of the original groups before they are joined.
\end{mytheo}

This paper provides Algorithm ~\ref{algo:cycle-graph-infinite} which is a way to generate arbitrarily large class of graph automorphism problems for which quantum Fourier sampling is guaranteed to fail.

\begin{algorithm}[H]
\label{algo:arb-fail}
\caption{An algorithm to create arbitrarily large easy graph automorphism problem}
     \label{algo:cycle-graph-infinite}
\begin{algorithmic}[1]
\Procedure {DISJOINT-CYCLE-GRAPH}{$m, n$} \Comment{$m$ is the number of nodes for the smallest cycle, $n$ is the number of cycles}

\State Create an $m$-cycle graph \Comment{The first cycle}

\For{$i\gets 1, n-1 $}
\State Increase $m$ by one
\State Create an $m$-cycle graph \Comment{The next cycle}
\EndFor
\EndProcedure
\end{algorithmic}
\end{algorithm}

By the end of its execution, Algorithm ~\ref{algo:cycle-graph-infinite} will generate a graph of $n$ cycles with a total of $m n$ nodes and the time complexity will be $O\left(poly(m, n)\right)$. One can be more creative about the Step 4 of Algorithm ~\ref{algo:cycle-graph-infinite} to create other classes of graph automorphism problems for which quantum Fourier sampling is guaranteed to fail.

At this point following remark can be made.

\begin{myrem}
\label{rem:dn-prod-g-qft-fail}
Arbitrarily large classes of graph automorphism problems can be created for which quantum Fourier sampling is guaranteed to fail. Quantum Fourier sampling is guaranteed to fail to compute the automorphism group of a graph when the automorphism group is the product of a dihedral group and any finite group.
\end{myrem}

\subsection{Existential approach}
In the existential approach, a general class of graphs is chosen and it is proven that there is at least one graph in that class for which quantum hidden subgroup algorithm is guaranteed to fail to compute the automorphism group.

The discussion starts with the Frucht's theorem \cite{frucht1939herstellung}.

\begin{mytheo}
Every abstract group is isomorphic to the automorphism group of some graph.
\end{mytheo}

So, any group which is a product of $D_n$ and a finite group $G$ is isomorphic to the automorphism group of some finite graph. It has also been shown in \cite{babai1996automorphism} that  every finite group as the group of symmetries of a strongly regular graph. It indicates that there is a class of strongly regular graph whose automorphism group is isomorphic to $D_n \times G$. According to the Theorem ~\ref{lab:irrep-prod-grp} and Remark ~\ref{rem:dn-prod-g-qft-fail}, one can argue that  quantum Fourier sampling should fail to construct the automorphism group of a subclass of strongly regular graphs.

Another example may be the  Cayley graph automorphism problem \cite{xu1998automorphism}. It is well known that the automorphism group of the Cayley graph $\text{Aut}\left( C \left(G, X\right)\right)$ of a group $G$ over a generating set $X$ contain an isomorphic copy of $G$ acting via left translations \cite{jajcay2000structure}. In that case, the automorphism group of the Cayley graph of a dihedral group $D_n$  contains $D_n$ as a subgroup. So, following the result of the previous section, quantum Fourier sampling fails to compute the automorphism group of $\text{Aut}\left( C \left(G, X\right)\right)$.

\subsection{Universal structures approach}
A class $\mathcal{C}$ of structures  is called {\it universal} if every finite group is the automorphism group of a structure in $\mathcal{C}$  \cite{cameron2004automorphisms}. A series of works by Frucht, Sabidussi, Mendelsohn, Babai, Kantor, and others \cite{cameron2004automorphisms} has shown the following classes of graphs to be universal - graphs of valency $k$ for any fixed $k > 2$ \cite{frucht1949graphs}; bipartite graphs; strongly regular graphs \cite{mendelsohn1978every}; Hamiltonian graphs \cite{sabidussi1957graphs}; $k$-connected graphs \cite{sabidussi1957graphs}, for $k > 0$; $k$-chromatic graphs, for $k > 1$; switching classes of graphs; lattices \cite{birkhoff1946grupos}; projective planes (possibly infinite); and Steiner triple systems \cite{mendelsohn1978groups}; and symmetric designs (BIBDs). It indicates that each of these classes has at least one graph which has its automorphism group isomorphic to $D_n \times G$ where $G$ is any finite group. So, quantum Fourier sampling will fail to compute the automorphism group of each of these cases.

\section{Is hidden subgroup the ideal approach?}
It has been shown that there are instances of graph automorphism problem for which hidden subgroup algorithm can never be successful although they are trivial to solve on a classical computer. So, it can be argued that the space of the hidden subgroup representations of all graph automorphism problems cannot capture the structure of the space of all graph automorphism problem. As the graph isomorphism problem is believed to be at least as hard as the graph automorphism problem, it can also be added that the space of the hidden subgroup representations of all graph isomorphism problem cannot capture the structure of the space of all graph isomorphism problem. So, it would be appropriate to investigate alternative quantum algorithmic approach for these classes of problems.

\section{Conclusion}
It has been shown that, while solving the hidden subgroup representation of the graph isomorphism problem is equivalent to determining order $2$ subgroup of a symmetric group,  the hidden subgroup representation of the graph automorphism problem is equivalent to determining a hidden subgroup of higher order. This paper has identified a class of graph automorphism problem for which the quantum Fourier transform algorithm always fails. It also has shown how one can determine non-trivial classes of graphs for which the same algorithm always fails. With these negative results, one may be interested to ask whether the hidden subgroup representation is a practical representation of the graph isomorphism and automorphism problems in quantum regime.

\section*{Acknowledgement}
OS thanks Dave Bacon, Aram Harrow, Robert Campbell, Marc Bogaerts, Andrew Childs,  Steven Gregory, Jef Laga, Dietrich Burde, Eric Wofsey, Alexander Hulpke, Michael Burr, Jyrki Lahtonen, Joshua Grochow and  Tobias Kildetoft for their helpful comments.

\newpage
\bibliographystyle{plain}
\bibliography{disst}

\newcommand{\noopsort}[1]{} \newcommand{\printfirst}[2]{#1}
  \newcommand{\singleletter}[1]{#1} \newcommand{\switchargs}[2]{#2#1}
\begin{thebibliography}{10}

\bibitem{Grochow2017}
Is graph automorphism karp-reducible to graph isomorphism under hidden subgroup
  representation?
\newblock
  \url{http://cstheory.stackexchange.com/questions/37635/is-graph-automorphism-karp-reducible-to-graph-isomorphism-under-hidden-subgroup}.
\newblock Accessed: 2017-03-02.

\bibitem{alagic2005strong}
Gorjan Alagic, Cristopher Moore, and Alexander Russell.
\newblock Strong fourier sampling fails over $g^{n}$.
\newblock {\em arXiv preprint quant-ph/0511054}, 2005.

\bibitem{JohnArmstrongCharInducRepre}
John Armstrong.
\newblock Characters of induced representations, 2010.
\newblock Accessed: 2016-05-25.

\bibitem{babai1996automorphism}
L{\'a}szl{\'o} Babai.
\newblock Automorphism groups, isomorphism, reconstruction.
\newblock In {\em Handbook of combinatorics (vol. 2)}, pages 1447--1540. MIT
  Press, 1996.

\bibitem{babai2015graph}
L{\'a}szl{\'o} Babai.
\newblock Graph isomorphism in quasipolynomial time.
\newblock {\em arXiv preprint arXiv:1512.03547}, 2015.

\bibitem{babai2013faster}
L{\'a}szl{\'o} Babai, Xi~Chen, Xiaorui Sun, Shang-Hua Teng, and John Wilmes.
\newblock Faster canonical forms for strongly regular graphs.
\newblock In {\em Foundations of Computer Science (FOCS), 2013 IEEE 54th Annual
  Symposium on}, pages 157--166. IEEE, 2013.

\bibitem{babai1980random}
L{\'a}szl{\'o} Babai, Paul Erd{\H{o}}s, and Stanley~M Selkow.
\newblock Random graph isomorphism.
\newblock {\em SIAM Journal on Computing}, 9(3):628--635, 1980.

\bibitem{babai1982isomorphism}
L{\'a}szl{\'o} Babai, D~Yu Grigoryev, and David~M Mount.
\newblock Isomorphism of graphs with bounded eigenvalue multiplicity.
\newblock In {\em Proceedings of the fourteenth annual ACM symposium on Theory
  of computing}, pages 310--324. ACM, 1982.

\bibitem{babai1983canonical}
L{\'a}szl{\'o} Babai and Eugene~M Luks.
\newblock Canonical labeling of graphs.
\newblock In {\em Proceedings of the fifteenth annual ACM symposium on Theory
  of computing}, pages 171--183. ACM, 1983.

\bibitem{bacon2005optimaldi}
Dave Bacon, Andrew~M Childs, and Wim van Dam.
\newblock Optimal measurements for the dihedral hidden subgroup problem.
\newblock {\em arXiv preprint quant-ph/0501044}, 2005.

\bibitem{birkhoff1946grupos}
Garrett Birkhoff.
\newblock Sobre los grupos de automorfismos.
\newblock {\em Rev. Uni{\'o}n Mat. Argent}, 11(4):155--157, 1946.

\bibitem{bollobas2013modern}
B{\'e}la Bollob{\'a}s.
\newblock {\em Modern graph theory}, volume 184.
\newblock Springer Science \& Business Media, 2013.

\bibitem{brassard1997exact}
Gilles Brassard and Peter Hoyer.
\newblock An exact quantum polynomial-time algorithm for simon's problem.
\newblock In {\em Theory of Computing and Systems, 1997., Proceedings of the
  Fifth Israeli Symposium on}, pages 12--23. IEEE, 1997.

\bibitem{cameron2004automorphisms}
Peter~J Cameron et~al.
\newblock Automorphisms of graphs.
\newblock {\em Topics in algebraic graph theory}, 102:137--155, 2004.

\bibitem{childs2016lecture}
Andrew~M Childs.
\newblock Lecture notes on quantum algorithms.
\newblock 2016.

\bibitem{RevModPhys.82.1}
Andrew~M. Childs and Wim van Dam.
\newblock Quantum algorithms for algebraic problems.
\newblock {\em Rev. Mod. Phys.}, 82:1--52, Jan 2010.

\bibitem{conrad2009dihedral}
KEITH Conrad.
\newblock Dihedral groups ii.
\newblock {\em Internet Online Book}, pages 3--6, 2009.

\bibitem{curtis1966representation}
Charles~W Curtis and Irving Reiner.
\newblock {\em Representation theory of finite groups and associative
  algebras}, volume 356.
\newblock American Mathematical Soc., 1966.

\bibitem{czajka2008improved}
Tomek Czajka and Gopal Pandurangan.
\newblock Improved random graph isomorphism.
\newblock {\em Journal of Discrete Algorithms}, 6(1):85--92, 2008.

\bibitem{diaconis1990efficient}
Persi Diaconis and Daniel Rockmore.
\newblock Efficient computation of the fourier transform on finite groups.
\newblock {\em Journal of the American Mathematical Society}, 3(2):297--332,
  1990.

\bibitem{dresselhaus2007group}
Mildred~S Dresselhaus, Gene Dresselhaus, and Ado Jorio.
\newblock {\em Group theory: application to the physics of condensed matter}.
\newblock Springer Science \& Business Media, 2007.

\bibitem{dummit2004abstract}
David~Steven Dummit and Richard~M Foote.
\newblock {\em Abstract algebra}, volume 1984.
\newblock Wiley Hoboken, 2004.

\bibitem{erdHos1963asymmetric}
Paul Erd{\H{o}}s and Alfr{\'e}d R{\'e}nyi.
\newblock Asymmetric graphs.
\newblock {\em Acta Mathematica Hungarica}, 14(3-4):295--315, 1963.

\bibitem{ettinger2000quantum}
Mark Ettinger and Peter H{\o}yer.
\newblock On quantum algorithms for noncommutative hidden subgroups.
\newblock {\em Advances in Applied Mathematics}, 25(3):239--251, 2000.

\bibitem{ettinger1999hidden}
Mark Ettinger, Peter Hoyer, and Emanuel Knill.
\newblock Hidden subgroup states are almost orthogonal.
\newblock {\em arXiv preprint quant-ph/9901034}, 1999.

\bibitem{ettinger2004quantum}
Mark Ettinger, Peter H{\o}yer, and Emanuel Knill.
\newblock The quantum query complexity of the hidden subgroup problem is
  polynomial.
\newblock {\em Information Processing Letters}, 91(1):43--48, 2004.

\bibitem{fortin1996graph}
Scott Fortin.
\newblock The graph isomorphism problem.
\newblock Technical report, Technical Report 96-20, University of Alberta,
  Edomonton, Alberta, Canada, 1996.

\bibitem{friedl2003hidden}
Katalin Friedl, G{\'a}bor Ivanyos, Fr{\'e}d{\'e}ric Magniez, Miklos Santha, and
  Pranab Sen.
\newblock Hidden translation and orbit coset in quantum computing.
\newblock In {\em Proceedings of the thirty-fifth annual ACM symposium on
  Theory of computing}, pages 1--9. ACM, 2003.

\bibitem{frucht1939herstellung}
Robert Frucht.
\newblock Herstellung von graphen mit vorgegebener abstrakter gruppe.
\newblock {\em Compositio Mathematica}, 6:239--250, 1939.

\bibitem{frucht1949graphs}
Robert Frucht.
\newblock Graphs of degree three with a given abstract group.
\newblock {\em Canadian J. Math}, 1:365--378, 1949.

\bibitem{fulton1991representation}
William Fulton and Joe Harris.
\newblock {\em Representation theory}, volume 129.
\newblock Springer Science \& Business Media, 1991.

\bibitem{ganesan2012automorphism}
Ashwin Ganesan.
\newblock Automorphism groups of graphs.
\newblock {\em arXiv preprint arXiv:1206.6279}, 2012.

\bibitem{garey2002computers}
Michael~R Garey and David~S Johnson.
\newblock {\em Computers and intractability}, volume~29.
\newblock wh freeman, 2002.

\bibitem{Gillespie2012}
Maria Gillespie.
\newblock Characters of the symmetric group.
\newblock
  \url{https://mathematicalgemstones.wordpress.com/2012/05/21/characters-of-the-symmetric-group/},
  2012.

\bibitem{grigni2001quantum}
Michelangelo Grigni, Leonard Schulman, Monica Vazirani, and Umesh Vazirani.
\newblock Quantum mechanical algorithms for the nonabelian hidden subgroup
  problem.
\newblock In {\em Proceedings of the thirty-third annual ACM symposium on
  Theory of computing}, pages 68--74. ACM, 2001.

\bibitem{hallgren2010limitations}
Sean Hallgren, Cristopher Moore, Martin R{\"o}tteler, Alexander Russell, and
  Pranab Sen.
\newblock Limitations of quantum coset states for graph isomorphism.
\newblock {\em Journal of the ACM (JACM)}, 57(6):34, 2010.

\bibitem{hallgren2000normal}
Sean Hallgren, Alexander Russell, and Amnon Ta-Shma.
\newblock Normal subgroup reconstruction and quantum computation using group
  representations.
\newblock In {\em Proceedings of the thirty-second annual ACM symposium on
  Theory of computing}, pages 627--635. ACM, 2000.

\bibitem{helfgott2017isomorphismes}
Harald~Andr{\'e}s Helfgott.
\newblock Isomorphismes de graphes en temps quasi-polynomial
  (d'apr$\backslash$es babai et luks, weisfeiler-leman...).
\newblock {\em arXiv preprint arXiv:1701.04372}, 2017.

\bibitem{holdener2009math}
Judy~A Holdener.
\newblock Math bite: Sums of sines and cosines.
\newblock {\em Mathematics Magazine}, 82(2):126--126, 2009.

\bibitem{holt2005handbook}
Derek~F Holt, Bettina Eick, and Eamonn~A O'Brien.
\newblock {\em Handbook of computational group theory}.
\newblock CRC Press, 2005.

\bibitem{hopcroft1974linear}
John~E Hopcroft and Jin-Kue Wong.
\newblock Linear time algorithm for isomorphism of planar graphs (preliminary
  report).
\newblock In {\em Proceedings of the sixth annual ACM symposium on Theory of
  computing}, pages 172--184. ACM, 1974.

\bibitem{hoyer2000quantum}
Peter H{\o}yer.
\newblock {\em Quantum Algorithms}.
\newblock PhD thesis, PhD thesis, Odense University, Denmark, 2000.

\bibitem{ivanyos2003efficient}
G{\'a}bor Ivanyos, Fr{\'e}d{\'e}ric Magniez, and Miklos Santha.
\newblock Efficient quantum algorithms for some instances of the non-abelian
  hidden subgroup problem.
\newblock {\em International Journal of Foundations of Computer Science},
  14(05):723--739, 2003.

\bibitem{jajcay2000structure}
Robert Jajcay.
\newblock The structure of automorphism groups of cayley graphs and maps.
\newblock {\em Journal of Algebraic Combinatorics}, 12(1):73--84, 2000.

\bibitem{jordan1869assemblages}
Camille Jordan.
\newblock Sur les assemblages de lignes.
\newblock {\em J. Reine Angew. Math}, 70(185):81, 1869.

\bibitem{jozsa1998quantum}
Richard Jozsa.
\newblock Quantum algorithms and the fourier transform.
\newblock In {\em Proceedings of the Royal Society of London A: Mathematical,
  Physical and Engineering Sciences}, volume 454, pages 323--337. The Royal
  Society, 1998.

\bibitem{lomonacopers2002}
Samuel J~Lomonaco Jr. and Louis~H Kauffman.
\newblock Quantum hidden subgroup algorithms: a mathematical perspective.
  quantum computation and information (washington, dc, 2000).
\newblock {\em Contemp. Math.}, 305:139–202, 2002.

\bibitem{knapp2009sines}
Michael~P Knapp.
\newblock Sines and cosines of angles in arithmetic progression.
\newblock {\em Mathematics Magazine}, 82(5):371--372, 2009.

\bibitem{knuth1991efficient}
Donald~E Knuth.
\newblock Efficient representation of perm groups.
\newblock {\em Combinatorica}, 11(1):33--43, 1991.

\bibitem{kobler2012graph}
Johannes Kobler, Uwe Sch{\"o}ning, and Jacobo Tor{\'a}n.
\newblock {\em The graph isomorphism problem: its structural complexity}.
\newblock Springer Science \& Business Media, 2012.

\bibitem{kuperberg2005subexponential}
Greg Kuperberg.
\newblock A subexponential-time quantum algorithm for the dihedral hidden
  subgroup problem.
\newblock {\em SIAM Journal on Computing}, 35(1):170--188, 2005.

\bibitem{luks1982isomorphism}
Eugene~M Luks.
\newblock Isomorphism of graphs of bounded valence can be tested in polynomial
  time.
\newblock {\em Journal of Computer and System Sciences}, 25(1):42--65, 1982.

\bibitem{mendelsohn1978every}
E~Mendelsohn.
\newblock Every (finite) group is the group of automorphisms of a (finite)
  strongly regular graph.
\newblock {\em Ars Combinatoria}, 6:75--86, 1978.

\bibitem{mendelsohn1978groups}
Eric Mendelsohn.
\newblock On the groups of automorphisms of steiner triple and quadruple
  systems.
\newblock {\em Journal of Combinatorial Theory, Series A}, 25(2):97--104, 1978.

\bibitem{meyer2000matrix}
Carl~D Meyer.
\newblock {\em Matrix analysis and applied linear algebra}, volume~2.
\newblock Siam, 2000.

\bibitem{miller1980isomorphism}
Gary Miller.
\newblock Isomorphism testing for graphs of bounded genus.
\newblock In {\em Proceedings of the twelfth annual ACM symposium on Theory of
  computing}, pages 225--235. ACM, 1980.

\bibitem{moore2006generic}
Cristopher Moore, Daniel Rockmore, and Alexander Russell.
\newblock Generic quantum fourier transforms.
\newblock {\em ACM Transactions on Algorithms (TALG)}, 2(4):707--723, 2006.

\bibitem{moore2002hidden}
Cristopher Moore, Daniel Rockmore, Alexander Russell, and Leonard Schulman.
\newblock The hidden subgroup problem in affine groups: Basis selection in
  fourier sampling.
\newblock {\em arXiv preprint quant-ph/0211124}, 2002.

\bibitem{moore2007power}
Cristopher Moore, Daniel Rockmore, Alexander Russell, and Leonard~J Schulman.
\newblock The power of strong fourier sampling: quantum algorithms for affine
  groups and hidden shifts.
\newblock {\em SIAM Journal on Computing}, 37(3):938--958, 2007.

\bibitem{moore2005explicit}
Cristopher Moore and Alexander Russell.
\newblock Explicit multiregister measurements for hidden subgroup problems.
\newblock {\em arXiv preprint quant-ph/0504067}, 2005.

\bibitem{moore2005distinguishing}
Cristopher Moore and Alexander Russell.
\newblock For distinguishing conjugate hidden subgroups, the pretty good
  measurement is as good as it gets.
\newblock {\em arXiv preprint quant-ph/0501177}, 2005.

\bibitem{moore2005tight}
Cristopher Moore and Alexander Russell.
\newblock Tight results on multiregister fourier sampling: Quantum measurements
  for graph isomorphism require entanglement.
\newblock {\em arXiv preprint quant-ph/0511149}, 2005.

\bibitem{moore2008symmetric}
Cristopher Moore, Alexander Russell, and Leonard~J Schulman.
\newblock The symmetric group defies strong fourier sampling.
\newblock {\em SIAM Journal on Computing}, 37(6):1842--1864, 2008.

\bibitem{moore2010impossibility}
Cristopher Moore, Alexander Russell, and Piotr Sniady.
\newblock On the impossibility of a quantum sieve algorithm for graph
  isomorphism.
\newblock {\em SIAM Journal on Computing}, 39(6):2377--2396, 2010.

\bibitem{mosca1999hidden}
Michele Mosca and Artur Ekert.
\newblock The hidden subgroup problem and eigenvalue estimation on a quantum
  computer.
\newblock In {\em Quantum Computing and Quantum Communications}, pages
  174--188. Springer, 1999.

\bibitem{radhakrishnan2005power}
Jaikumar Radhakrishnan, Martin R{\"o}tteler, and Pranab Sen.
\newblock On the power of random bases in fourier sampling: Hidden subgroup
  problem in the heisenberg group.
\newblock In {\em International Colloquium on Automata, Languages, and
  Programming}, pages 1399--1411. Springer, 2005.

\bibitem{Read1977}
Ronald~C. Read and Derek~G. Corneil.
\newblock The graph isomorphism disease.
\newblock {\em Journal of Graph Theory}, 1(4):339--363, 1977.

\bibitem{roetteler1998polynomial}
Martin Roetteler and Thomas Beth.
\newblock Polynomial-time solution to the hidden subgroup problem for a class
  of non-abelian groups.
\newblock Technical report, 1998.

\bibitem{sabidussi1957graphs}
Gert Sabidussi.
\newblock Graphs with given group and given graph-theoretical properties.
\newblock {\em Canad. J. Math}, 9(515):C525, 1957.

\bibitem{sagan2013symmetric}
Bruce Sagan.
\newblock {\em The symmetric group: representations, combinatorial algorithms,
  and symmetric functions}, volume 203.
\newblock Springer Science \& Business Media, 2013.

\bibitem{schoning1988graph}
Uwe Sch{\"o}ning.
\newblock Graph isomorphism is in the low hierarchy.
\newblock {\em Journal of Computer and System Sciences}, 37(3):312--323, 1988.

\bibitem{seress2003permutation}
{\'A}kos Seress.
\newblock {\em Permutation group algorithms}, volume 152.
\newblock Cambridge University Press, 2003.

\bibitem{serre2012linear}
Jean-Pierre Serre.
\newblock {\em Linear representations of finite groups}, volume~42.
\newblock Springer Science \& Business Media, 2012.

\bibitem{shor1999polynomial}
Peter~W Shor.
\newblock Polynomial-time algorithms for prime factorization and discrete
  logarithms on a quantum computer.
\newblock {\em SIAM review}, 41(2):303--332, 1999.

\bibitem{sims1970computational}
Charles~C Sims.
\newblock Computational methods in the study of permutation groups.
\newblock In {\em Computational problems in abstract algebra}, pages 169--183,
  1970.

\bibitem{stanley1986enumerative}
Richard~P Stanley.
\newblock {\em What Is Enumerative Combinatorics?}
\newblock Springer, 1986.

\bibitem{van2012quantum}
Wim van Dam and Yoshitaka Sasaki.
\newblock Quantum algorithms for problems in number theory, algebraic geometry,
  and group theory.
\newblock {\em Diversities in Quantum Computation and Quantum Information},
  5:79, 2012.

\bibitem{xu1998automorphism}
Ming-Yao Xu.
\newblock Automorphism groups and isomorphisms of cayley digraphs.
\newblock {\em Discrete Mathematics}, 182(1):309--319, 1998.

\end{thebibliography}

\newpage
\appendix
\section{Irreducible representations of $\text{Aut} \left(C_m \sqcup C_n\right)$}
\label{app:cmcn}
\begin{enumerate}
  \item $\rho_1 = \rho_{1,1, m} \otimes \rho_{1,1, n} = 1 \times 1  = 1$. So, the character is $\chi_1 = 1$.
  \item $\rho_2 = \rho_{2,1, m} \otimes \rho_{1,1, n} = 1 \times 1  \text{ or } -1 \times 1  = \pm 1$. So, the character is  $\chi_2 = \pm 1$.
  \item $\rho_3 = \rho_{3,1, m} \otimes \rho_{1,1, n} = 1 \times 1  \text{ or } -1 \times 1  = \pm 1  $. So, the character is   $\chi_3 = \pm 1$.
  \item $\rho_4 = \rho_{4,1, m} \otimes \rho_{1,1, n} = 1 \times 1  \text{ or } -1 \times 1  = \pm 1  $. So, the character is   $\chi_4 = \pm 1$.
  \item $\rho_5 = \rho_{1,2, m} \otimes \rho_{1,1, n} = \begin{pmatrix}
e^{\frac{2 \pi i k_m}{m}}&0\\
0&e^{-\frac{2 \pi i k_m}{m}}
\end{pmatrix} \times 1  = \begin{pmatrix}
e^{\frac{2 \pi i k_m}{m}}&0\\
0&e^{-\frac{2 \pi i k_m}{m}}
\end{pmatrix}$. So, the character is $\chi_5 =e^{\frac{2 \pi i k_m}{m}} + e^{-\frac{2 \pi i k_m}{m}}$ $= 2 \cos \left(\frac{2 \pi k_m}{m}\right)$.
  \item $\rho_6 = \rho_{2,2, m} \otimes \rho_{1,1, n} = \begin{pmatrix}
e^{\frac{2 \pi i k_m l_m}{m}}&0\\
0&e^{-\frac{2 \pi i k_m l_m}{m}}
\end{pmatrix} \times 1 = \begin{pmatrix}
e^{\frac{2 \pi i k_m l_m}{m}}&0\\
0&e^{-\frac{2 \pi i k_m l_m}{m}}
\end{pmatrix}$. So, the character is $\chi_6 = e^{\frac{2 \pi i k_m l_m}{m}} + e^{-\frac{2 \pi i k_m l_m}{m}}$ $= 2 \cos \left(\frac{2 \pi k_m l_m}{m}\right)$.
  \item $\rho_7 = \rho_{3,2, m} \otimes \rho_{1,1, n} = \begin{pmatrix}
0&1\\
1&0
\end{pmatrix} \times 1 = \begin{pmatrix}
0&1\\
1&0
\end{pmatrix}$. So, the character is $\chi_7 = 0$.
  \item $\rho_8 = \rho_{4,2, m} \otimes \rho_{1,1, n} = \begin{pmatrix}
0&e^{\frac{2 \pi i k_m l_m}{m}}\\
e^{-\frac{2 \pi i k_m l_m}{m}}&0
\end{pmatrix} \times 1 = \begin{pmatrix}
0&e^{\frac{2 \pi i k_m l_m}{m}}\\
e^{-\frac{2 \pi i k_m l_m}{m}}&0
\end{pmatrix}$. So, the character is $\chi_8 = 0$.
  \item $\rho_9 = \rho_{1,1, m} \otimes \rho_{2,1, n} = 1 \times 1 \text{ or } 1 \times -1  = \pm 1$. So, the character is $\chi_9 = \pm 1$.
  \item $\rho_{10} = \rho_{2,1, m} \otimes \rho_{2,1, n} = 1 \times 1 \text{ or } -1 \times1 \text{ or } 1 \times -1 \text{ or } -1 \times -1   = \pm 1$. So, the character is $\chi_{10} = \pm 1$.
  \item $\rho_{11} = \rho_{3,1, m} \otimes \rho_{2,1, n} = 1 \times 1 \text{ or } -1 \times1 \text{ or } 1 \times -1 \text{ or } -1 \times -1   = \pm 1$. So, the character is $\chi_{11} = \pm 1$.
  \item $\rho_{12} = \rho_{4,1, m} \otimes \rho_{2,1, n} = 1 \times 1 \text{ or } -1 \times1 \text{ or } 1 \times -1 \text{ or } -1 \times -1   = \pm 1$. So, the character is $\chi_{12} = \pm 1$.
  \item $\rho_{13} = \rho_{1,2, m} \otimes \rho_{2,1, n} = \begin{pmatrix}
e^{\frac{2 \pi i k_m}{m}}&0\\
0&e^{-\frac{2 \pi i k_m}{m}}
\end{pmatrix} \times 1 \text{ or } \begin{pmatrix}
e^{\frac{2 \pi i k_m}{m}}&0\\
0&e^{-\frac{2 \pi i k_m}{m}}
\end{pmatrix} \times -1$ \\$= \pm \begin{pmatrix}
e^{\frac{2 \pi i k_m}{m}}&0\\
0&e^{-\frac{2 \pi i k_m}{m}}
\end{pmatrix}$. So, the character is $\chi_{13} =\pm \left( e^{\frac{2 \pi i k_m}{m}} + e^{-\frac{2 \pi i k_m}{m}}\right)$\\ $= \pm 2 \cos \left(\frac{2 \pi k_m}{m}\right)$.
  \item $\rho_{14} = \rho_{2,2, m} \otimes \rho_{2,1, n} = \begin{pmatrix}
e^{\frac{2 \pi i k_m l_m}{m}}&0\\
0&e^{-\frac{2 \pi i k_m l_m}{m}}
\end{pmatrix} \times 1 \text{ or } \begin{pmatrix}
e^{\frac{2 \pi i k_m l_m}{m}}&0\\
0&e^{-\frac{2 \pi i k_m l_m}{m}}
\end{pmatrix} \times -1$\\ $= \pm \begin{pmatrix}
e^{\frac{2 \pi i k_m l_m}{m}}&0\\
0&e^{-\frac{2 \pi i k_m l_m}{m}}
\end{pmatrix}$. So, the character is $\chi_{14} =\pm \left(e^{\frac{2 \pi i k_m l_m}{m}} + e^{-\frac{2 \pi i k_m l_m}{m}} \right)$\\ $ = \pm  2 \cos \left(\frac{2 \pi k_m l_m}{m}\right)$.
  \item $\rho_{15} = \rho_{3,2, m} \otimes \rho_{2,1, n} = \begin{pmatrix}
0&1\\
1&0
\end{pmatrix} \times 1 \text{ or } \begin{pmatrix}
0&1\\
1&0
\end{pmatrix} \times -1$  $ = \pm \begin{pmatrix}
0&1\\
1&0
\end{pmatrix}$. So, the character is $\chi_{15} =0$.
  \item $\rho_{16} = \rho_{4,2, m} \otimes \rho_{2,1, n} = \begin{pmatrix}
0&e^{\frac{2 \pi i k_m l_m}{m}}\\
e^{-\frac{2 \pi i k_m l_m}{m}}&0
\end{pmatrix} \times 1 \text{ or } \begin{pmatrix}
0&e^{\frac{2 \pi i k_m l_m}{m}}\\
e^{-\frac{2 \pi i k_m l_m}{m}}&0
\end{pmatrix} \times -1$ \\   $ = \pm \begin{pmatrix}
0&e^{\frac{2 \pi i k_m l_m}{m}}\\
e^{-\frac{2 \pi i k_m l_m}{m}}&0
\end{pmatrix}$. So, the character is $\chi_{16} =0$.
  \item $\rho_{17} = \rho_{1,1, m} \otimes \rho_{3,1, n} = 1 \times 1 \text{ or } 1 \times -1 = \pm 1$. So, the character is $\chi_{17} =\pm 1$.
  \item $\rho_{18} = \rho_{2,1, m} \otimes \rho_{3,1, n} = 1 \times 1 \text{ or } -1 \times1 \text{ or } 1 \times -1 \text{ or } -1 \times -1  = \pm 1$. So, the character is $\chi_{18} =\pm 1$.
  \item $\rho_{19} = \rho_{3,1, m} \otimes \rho_{3,1, n} = 1 \times 1 \text{ or } -1 \times1 \text{ or } 1 \times -1 \text{ or } -1 \times -1  = \pm 1$. So, the character is $\chi_{19} =\pm 1$.
  \item $\rho_{20} = \rho_{4,1, m} \otimes \rho_{3,1, n} = 1 \times 1 \text{ or } -1 \times1 \text{ or } 1 \times -1 \text{ or } -1 \times -1  = \pm 1$. So, the character is $\chi_{20} =\pm 1$.
  \item $\rho_{21} = \rho_{1,2, m} \otimes \rho_{3,1, n} = \begin{pmatrix}
e^{\frac{2 \pi i k_m}{m}}&0\\
0&e^{-\frac{2 \pi i k_m}{m}}
\end{pmatrix} \times 1 \text{ or } \begin{pmatrix}
e^{\frac{2 \pi i k_m}{m}}&0\\
0&e^{-\frac{2 \pi i k_m}{m}}
\end{pmatrix} \times -1$ $ = \pm \begin{pmatrix}
e^{\frac{2 \pi i k_m}{m}}&0\\
0&e^{-\frac{2 \pi i k_m}{m}}
\end{pmatrix}$. So, the character is $\chi_{21} =\pm \left(e^{\frac{2 \pi i k_m}{m}} + e^{-\frac{2 \pi i k_m}{m}} \right)$ $= \pm 2 \cos \left(\frac{2 \pi k_m}{m}\right)$.
  \item $\rho_{22} = \rho_{2,2, m} \otimes \rho_{3,1, n} = \begin{pmatrix}
e^{\frac{2 \pi i k_m l_m}{m}}&0\\
0&e^{-\frac{2 \pi i k_m l_m}{m}}
\end{pmatrix} \times 1 \text{ or } \begin{pmatrix}
e^{\frac{2 \pi i k_m l_m}{m}}&0\\
0&e^{-\frac{2 \pi i k_m l_m}{m}}
\end{pmatrix} \times -1$\\ $ = \pm \begin{pmatrix}
e^{\frac{2 \pi i k_m l_m}{m}}&0\\
0&e^{-\frac{2 \pi i k_m l_m}{m}}
\end{pmatrix}$. So, the character is $ \chi_{22} = \pm \left(e^{\frac{2 \pi i k_m l_m}{m}} + e^{-\frac{2 \pi i k_m l_m}{m}} \right)$ $= 2 \cos \left(\frac{2 \pi k_m l_m}{m}\right)$.
  \item $\rho_{23} = \rho_{3,2, m} \otimes \rho_{3,1, n} = \begin{pmatrix}
0&1\\
1&0
\end{pmatrix} \times 1 \text{ or } \begin{pmatrix}
0&1\\
1&0
\end{pmatrix} \times -1$  $ = \pm \begin{pmatrix}
0&1\\
1&0
\end{pmatrix}$. So, the character is $\chi_{23} =0$.
  \item $\rho_{24} = \rho_{4,2, m} \otimes \rho_{3,1, n} = \begin{pmatrix}
0&e^{\frac{2 \pi i k_m l_m}{m}}\\
e^{-\frac{2 \pi i k_m l_m}{m}}&0
\end{pmatrix} \times 1 \text{ or } \begin{pmatrix}
0&e^{\frac{2 \pi i k_m l_m}{m}}\\
e^{-\frac{2 \pi i k_m l_m}{m}}&0
\end{pmatrix} \times -1$ \\     $ = \pm \begin{pmatrix}
0&e^{\frac{2 \pi i k_m l_m}{m}}\\
e^{-\frac{2 \pi i k_m l_m}{m}}&0
\end{pmatrix}$. So, the character is $\chi_{24} =0$.
  \item $\rho_{25} = \rho_{1,1, m} \otimes \rho_{4,1, n} = 1 \times 1 \text{ or } 1 \times -1 = \pm 1$. So, the character is $\chi_{25} =\pm 1$.
  \item $\rho_{26} = \rho_{2,1, m} \otimes \rho_{4,1, n} = 1 \times 1 \text{ or } -1 \times1 \text{ or } 1 \times -1 \text{ or } -1 \times -1  = \pm 1$. So, the character is  $\chi_{26} =\pm 1$.
  \item $\rho_{27} = \rho_{3,1, m} \otimes \rho_{4,1, n} = 1 \times 1 \text{ or } -1 \times1 \text{ or } 1 \times -1 \text{ or } -1 \times -1  = \pm 1$. So, the character is  $\chi_{27} =\pm 1$.
  \item $\rho_{28} = \rho_{4,1, m} \otimes \rho_{4,1, n} = 1 \times 1 \text{ or } -1 \times1 \text{ or } 1 \times -1 \text{ or } -1 \times -1  = \pm 1$. So, the character is  $\chi_{28} =\pm 1$.
  \item $\rho_{29} = \rho_{1,2, m} \otimes \rho_{4,1, n} = \begin{pmatrix}
e^{\frac{2 \pi i k_m}{m}}&0\\
0&e^{-\frac{2 \pi i k_m}{m}}
\end{pmatrix} \times 1 \text{ or } \begin{pmatrix}
e^{\frac{2 \pi i k_m}{m}}&0\\
0&e^{-\frac{2 \pi i k_m}{m}}
\end{pmatrix} \times -1$ $ = \pm \begin{pmatrix}
e^{\frac{2 \pi i k_m}{m}}&0\\
0&e^{-\frac{2 \pi i k_m}{m}}
\end{pmatrix}$. So, the character is $\chi_{29} =\pm \left(e^{\frac{2 \pi i k_m}{m}} + e^{-\frac{2 \pi i k_m}{m}} \right)$ $= \pm 2 \cos \left(\frac{2 \pi k_m}{m}\right)$.
  \item $\rho_{30} = \rho_{2,2, m} \otimes \rho_{4,1, n} = \begin{pmatrix}
e^{\frac{2 \pi i k_m l_m}{m}}&0\\
0&e^{-\frac{2 \pi i k_m l_m}{m}}
\end{pmatrix} \times 1 \text{ or } \begin{pmatrix}
e^{\frac{2 \pi i k_m l_m}{m}}&0\\
0&e^{-\frac{2 \pi i k_m l_m}{m}}
\end{pmatrix} \times -1$\\ $ = \pm \begin{pmatrix}
e^{\frac{2 \pi i k_m l_m}{m}}&0\\
0&e^{-\frac{2 \pi i k_m l_m}{m}}
\end{pmatrix}$. So, the character is $\chi_{30} =\pm \left(e^{\frac{2 \pi i k_m l_m}{m}} + e^{-\frac{2 \pi i k_m l_m}{m}} \right)$\\ $=\pm 2 \cos \left(\frac{2\pi k_m l_m}{m}\right)$.
  \item $\rho_{31} = \rho_{3,2, m} \otimes \rho_{4,1, n} = \begin{pmatrix}
0&1\\
1&0
\end{pmatrix} \times 1 \text{ or } \begin{pmatrix}
0&1\\
1&0
\end{pmatrix} \times -1$  $ = \pm \begin{pmatrix}
0&1\\
1&0
\end{pmatrix}$. So, the character is $\chi_{31} =0$.
  \item $\rho_{32} = \rho_{4,2, m} \otimes \rho_{4,1, n} = \begin{pmatrix}
0&e^{\frac{2 \pi i k_m l_m}{m}}\\
e^{-\frac{2 \pi i k_m l_m}{m}}&0
\end{pmatrix} \times 1 \text{ or } \begin{pmatrix}
0&e^{\frac{2 \pi i k_m l_m}{m}}\\
e^{-\frac{2 \pi i k_m l_m}{m}}&0
\end{pmatrix} \times -1$  \\ $ = \pm \begin{pmatrix}
0&e^{\frac{2 \pi i k_m l_m}{m}}\\
e^{-\frac{2 \pi i k_m l_m}{m}}&0
\end{pmatrix}$. So, the character is $\chi_{32} =0$.     
  \item $\rho_{33} = \rho_{1,1, m} \otimes \rho_{1,2, n} = 1 \times \begin{pmatrix}
e^{\frac{2 \pi i k_n}{n}}&0\\
0&e^{-\frac{2 \pi i k_n}{n}}
\end{pmatrix} = \begin{pmatrix}
e^{\frac{2 \pi i k_n}{n}}&0\\
0&e^{-\frac{2 \pi i k_n}{n}}
\end{pmatrix}$.\\ So, the character is $\chi_{33} =\left(e^{\frac{2 \pi i k_n}{n}} + e^{-\frac{2 \pi i k_n}{n}} \right)$ $= 2 \cos \left(\frac{2 \pi k_n}{n}\right)$.
  \item $\rho_{34} = \rho_{2,1, m} \otimes \rho_{1,2, n} = 1 \times \begin{pmatrix}
e^{\frac{2 \pi i k_n}{n}}&0\\
0&e^{-\frac{2 \pi i k_n}{n}}
\end{pmatrix}  \text{ or } -1 \times  \begin{pmatrix}
e^{\frac{2 \pi i k_n}{n}}&0\\
0&e^{-\frac{2 \pi i k_n}{n}}
\end{pmatrix}$ $ = \pm \begin{pmatrix}
e^{\frac{2 \pi i k_n}{n}}&0\\
0&e^{-\frac{2 \pi i k_n}{n}}
\end{pmatrix}$. So, the character is $\chi_{34} =\pm \left(e^{\frac{2 \pi i k_n}{n}} + e^{-\frac{2 \pi i k_n}{n}} \right)$ $=\pm 2 \cos \left(\frac{2 \pi k_n}{n}\right)$.
  \item $\rho_{35} = \rho_{3,1, m} \otimes \rho_{1,2, n} = 1 \times \begin{pmatrix}
e^{\frac{2 \pi i k_n}{n}}&0\\
0&e^{-\frac{2 \pi i k_n}{n}}
\end{pmatrix}  \text{ or } -1 \times  \begin{pmatrix}
e^{\frac{2 \pi i k_n}{n}}&0\\
0&e^{-\frac{2 \pi i k_n}{n}}
\end{pmatrix}$ $ = \pm \begin{pmatrix}
e^{\frac{2 \pi i k_n}{n}}&0\\
0&e^{-\frac{2 \pi i k_n}{n}}
\end{pmatrix}$. So, the character is $\chi_{35} =\pm \left(e^{\frac{2 \pi i k_n}{n}} + e^{-\frac{2 \pi i k_n}{n}} \right)$ $=\pm 2 \cos \left(\frac{2 \pi k_n}{n}\right)$.
  \item $\rho_{36} = \rho_{4,1, m} \otimes \rho_{1,2, n} = 1 \times \begin{pmatrix}
e^{\frac{2 \pi i k_n}{n}}&0\\
0&e^{-\frac{2 \pi i k_n}{n}}
\end{pmatrix}  \text{ or } -1 \times  \begin{pmatrix}
e^{\frac{2 \pi i k_n}{n}}&0\\
0&e^{-\frac{2 \pi i k_n}{n}}
\end{pmatrix}$ $ = \pm \begin{pmatrix}
e^{\frac{2 \pi i k_n}{n}}&0\\
0&e^{-\frac{2 \pi i k_n}{n}}
\end{pmatrix}$. So, the character is $\chi_{36} =\pm \left(e^{\frac{2 \pi i k_n}{n}} + e^{-\frac{2 \pi i k_n}{n}} \right)$ $=\pm 2 \cos \left(\frac{2 \pi k_n}{n}\right)$.
  \item $\rho_{37} = \rho_{1,2, m} \otimes \rho_{1,2, n} = \begin{pmatrix}
e^{\frac{2 \pi i k_m}{m}}&0\\
0&e^{-\frac{2 \pi i k_m}{m}}
\end{pmatrix} \otimes \begin{pmatrix}
e^{\frac{2 \pi i k_n}{n}}&0\\
0&e^{-\frac{2 \pi i k_n}{n}}
\end{pmatrix}$ \\$ = \left(
\begin{array}{cccc}
 e^{\frac{2 i \pi  k_m}{m}+\frac{2 i k_n \pi }{n}} & 0 & 0 & 0 \\
 0 & e^{\frac{2 i k_m \pi }{m}-\frac{2 i k_n \pi }{n}} & 0 & 0 \\
 0 & 0 & e^{\frac{2 i k_n \pi }{n}-\frac{2 i k_m \pi }{m}} & 0 \\
 0 & 0 & 0 & e^{-\frac{2 i \pi  k_m}{m}-\frac{2 i k_n \pi }{n}} \\
\end{array}
\right)$.\\ So, the character is $\chi_{37} = 4 \cos \left(\frac{2 \pi  k_m}{m}\right) \cos \left(\frac{2 \pi  k_n}{n}\right)$.
  \item $\rho_{38} = \rho_{2,2, m} \otimes \rho_{1,2, n} = \begin{pmatrix}
e^{\frac{2 \pi i k_m l_m}{m}}&0\\
0&e^{-\frac{2 \pi i k_m l_m}{m}}
\end{pmatrix} \otimes \begin{pmatrix}
e^{\frac{2 \pi i k_n}{n}}&0\\
0&e^{-\frac{2 \pi i k_n}{n}}
\end{pmatrix}$ \\ $= \left(
\begin{array}{cccc}
 e^{\frac{2 i \pi  k_n}{n}+\frac{2 i k_m l_m \pi }{m}} & 0 & 0 & 0 \\
 0 & e^{\frac{2 i k_m l_m \pi }{m}-\frac{2 i k_n \pi }{n}} & 0 & 0 \\
 0 & 0 & e^{\frac{2 i k_n \pi }{n}-\frac{2 i k_m l_m \pi }{m}} & 0 \\
 0 & 0 & 0 & e^{-\frac{2 i \pi  k_n}{n}-\frac{2 i k_m l_m \pi }{m}} \\
\end{array}
\right)$.\\ So, the character is $\chi_{38} = 4 \cos \left(\frac{2 \pi  k_n}{n}\right) \cos \left(\frac{2 \pi  k_m l_m}{m}\right)$.
  \item $\rho_{39} = \rho_{3,2, m} \otimes \rho_{1,2, n} = \begin{pmatrix}
0&1\\
1&0
\end{pmatrix} \otimes \begin{pmatrix}
e^{\frac{2 \pi i k_n}{n}}&0\\
0&e^{-\frac{2 \pi i k_n}{n}}
\end{pmatrix}$ \\ $ = \left(
\begin{array}{cccc}
 0 & 0 & e^{\frac{2 i k_n \pi }{n}} & 0 \\
 0 & 0 & 0 & e^{-2 \frac{1}{n} i k_n \pi } \\
 e^{\frac{2 i k_n \pi }{n}} & 0 & 0 & 0 \\
 0 & e^{- 2 \frac{1}{n} i k_n \pi } & 0 & 0 \\
\end{array}
\right)$. So, the character is $\chi_{39} =0$.
  \item $\rho_{40} = \rho_{4,2, m} \otimes \rho_{1,2, n} = \begin{pmatrix}
0&e^{\frac{2 \pi i k_m l_m}{m}}\\
e^{-\frac{2 \pi i k_m l_m}{m}}&0
\end{pmatrix} \otimes \begin{pmatrix}
e^{\frac{2 \pi i k_n}{n}}&0\\
0&e^{-\frac{2 \pi i k_n}{n}}
\end{pmatrix}$  \\ $ = \left(
\begin{array}{cccc}
 0 & 0 & e^{\frac{2 i \pi  k_n}{n}+\frac{2 i k_m l_m \pi }{m}} & 0 \\
 0 & 0 & 0 & e^{\frac{2 i k_m l_m \pi }{m}-\frac{2 i k_n \pi }{n}} \\
 e^{\frac{2 i k_n \pi }{n}-\frac{2 i k_m lm \pi }{m}} & 0 & 0 & 0 \\
 0 & e^{-\frac{2 i \pi  k_n}{n}-\frac{2 i k_m l_m \pi }{m}} & 0 & 0 \\
\end{array}
\right)$. So, the character is $\chi_{40} =0$.
  \item $\rho_{41} = \rho_{1,1, m} \otimes \rho_{2,2, n} = 1 \times \begin{pmatrix}
e^{\frac{2 \pi i k_n l_n}{n}}&0\\
0&e^{-\frac{2 \pi i k_n l_n}{n}}
\end{pmatrix} = \begin{pmatrix}
e^{\frac{2 \pi i k_n l_n}{n}}&0\\
0&e^{-\frac{2 \pi i k_n l_n}{n}}
\end{pmatrix}$.\\ So, the character is $\chi_{41} =\left(e^{\frac{2 \pi i k_n l_n}{n}} + e^{-\frac{2 \pi i k_n l_n}{n}} \right)$ $= 2 \cos \left(\frac{2 \pi k_n l_n}{n}\right)$.
  \item $\rho_{42} = \rho_{2,1, m} \otimes \rho_{2,2, n} = 1 \times \begin{pmatrix}
e^{\frac{2 \pi i k_n l_n}{n}}&0\\
0&e^{-\frac{2 \pi i k_n l_n}{n}}
\end{pmatrix} \text{ or } -1 \times \begin{pmatrix}
e^{\frac{2 \pi i k_n l_n}{n}}&0\\
0&e^{-\frac{2 \pi i k_n l_n}{n}}
\end{pmatrix} $\\ $ = \pm \begin{pmatrix}
e^{\frac{2 \pi i k_n l_n}{n}}&0\\
0&e^{-\frac{2 \pi i k_n l_n}{n}}
\end{pmatrix}$. So, the character is $\chi_{42} =\pm \left(e^{\frac{2 \pi i k_n l_n}{n}} + e^{-\frac{2 \pi i k_n l_n}{n}} \right)$\\ $=\pm 2 \cos \left(\frac{2 \pi k_n l_n}{n}\right)$.
  \item $\rho_{43} = \rho_{3,1, m} \otimes \rho_{2,2, n} = 1 \times \begin{pmatrix}
e^{\frac{2 \pi i k_n l_n}{n}}&0\\
0&e^{-\frac{2 \pi i k_n l_n}{n}}
\end{pmatrix} \text{ or } -1 \times \begin{pmatrix}
e^{\frac{2 \pi i k_n l_n}{n}}&0\\
0&e^{-\frac{2 \pi i k_n l_n}{n}}
\end{pmatrix} $\\ $ = \pm \begin{pmatrix}
e^{\frac{2 \pi i k_n l_n}{n}}&0\\
0&e^{-\frac{2 \pi i k_n l_n}{n}}
\end{pmatrix}$. So, the character is $\chi_{43} =\pm \left(e^{\frac{2 \pi i k_n l_n}{n}} + e^{-\frac{2 \pi i k_n l_n}{n}} \right)$ \\$=\pm 2 \cos \left(\frac{2 \pi k_n l_n}{n}\right)$.
  \item $\rho_{44} = \rho_{4,1, m} \otimes \rho_{2,2, n} = 1 \times \begin{pmatrix}
e^{\frac{2 \pi i k_n l_n}{n}}&0\\
0&e^{-\frac{2 \pi i k_n l_n}{n}}
\end{pmatrix} \text{ or } -1 \times \begin{pmatrix}
e^{\frac{2 \pi i k_n l_n}{n}}&0\\
0&e^{-\frac{2 \pi i k_n l_n}{n}}
\end{pmatrix} $\\ $ = \pm \begin{pmatrix}
e^{\frac{2 \pi i k_n l_n}{n}}&0\\
0&e^{-\frac{2 \pi i k_n l_n}{n}}
\end{pmatrix}$. So, the character is $\chi_{44} =\pm \left(e^{\frac{2 \pi i k_n l_n}{n}} +e^{-\frac{2 \pi i k_n l_n}{n}} \right)$\\ $=\pm 2 \cos \left(\frac{2 \pi k_n l_n}{n}\right)$.
  \item $\rho_{45} = \rho_{1,2, m} \otimes \rho_{2,2, n} = \begin{pmatrix}
e^{\frac{2 \pi i k_m}{m}}&0\\
0&e^{-\frac{2 \pi i k_m}{m}}
\end{pmatrix} \otimes \begin{pmatrix}
e^{\frac{2 \pi i k_n l_n}{n}}&0\\
0&e^{-\frac{2 \pi i k_n l_n}{n}}
\end{pmatrix}$ \\ $ = \left(
\begin{array}{cccc}
 e^{\frac{2 i \pi  k_m}{m}+\frac{2 i k_n  l_n  \pi }{n}} & 0 & 0 & 0 \\
 0 & e^{\frac{2 i k_m \pi }{m}-\frac{2 i k_n  l_n  \pi }{n}} & 0 & 0 \\
 0 & 0 & e^{\frac{2 i k_n  l_n  \pi }{n}-\frac{2 i k_m \pi }{m}} & 0 \\
 0 & 0 & 0 & e^{-\frac{2 i \pi  k_m}{m}-\frac{2 i k_n  l_n  \pi }{n}} \\
\end{array}
\right)$.\\ So, the character is $\chi_{45} = 4 \cos \left(\frac{2 \pi  k_m}{m}\right) \cos \left(\frac{2 \pi  k_n l_n }{n}\right)$.
  \item $\rho_{46} = \rho_{2,2, m} \otimes \rho_{2,2, n} = \begin{pmatrix}
e^{\frac{2 \pi i k_m l_m}{m}}&0\\
0&e^{-\frac{2 \pi i k_m l_m}{m}}
\end{pmatrix} \otimes \begin{pmatrix}
e^{\frac{2 \pi i k_n l_n}{n}}&0\\
0&e^{-\frac{2 \pi i k_n l_n}{n}}
\end{pmatrix}$ \\ $ = \left(
\begin{array}{cccc}
 e^{\frac{2 i k_m \pi  l_m}{m}+\frac{2 i k_n l_n  \pi }{n}} & 0 & 0 & 0 \\
 0 & e^{\frac{2 i k_m l_m \pi }{m}-\frac{2 i k_n l_n  \pi }{n}} & 0 & 0 \\
 0 & 0 & e^{\frac{2 i k_n  l_n  \pi }{n}-\frac{2 i k_m l_m \pi }{m}} & 0 \\
 0 & 0 & 0 & e^{-\frac{2 i k_m \pi  l_m}{m}-\frac{2 i k_n  l_n  \pi }{n}} \\
\end{array}
\right)$.\\ So, the character is $\chi_{46} = 4 \cos \left(\frac{2 \pi  k_m l_m}{m}\right) \cos \left(\frac{2 \pi  k_n l_n }{n}\right)$. 
  \item $\rho_{47} = \rho_{3,2, m} \otimes \rho_{2,2, n} = \begin{pmatrix}
0&1\\
1&0
\end{pmatrix} \otimes \begin{pmatrix}
e^{\frac{2 \pi i k_n l_n}{n}}&0\\
0&e^{-\frac{2 \pi i k_n l_n}{n}}
\end{pmatrix}$  \\ $ = \left(
\begin{array}{cccc}
 0 & 0 & e^{\frac{1}{n} 2 i k_n l_n \pi } & 0 \\
 0 & 0 & 0 & e^{-\frac{1}{n} 2 i k_n l_n \pi } \\
 e^{\frac{1}{n} 2 i k_n l_n \pi } & 0 & 0 & 0 \\
 0 & e^{-\frac{1}{n} 2 i k_n l_n \pi } & 0 & 0 \\
\end{array}
\right)$. So, the character is $\chi_{47} =0$.
  \item $\rho_{48} = \rho_{4,2, m} \otimes \rho_{2,2, n} = \begin{pmatrix}
0&e^{\frac{2 \pi i k_m l_m}{m}}\\
e^{-\frac{2 \pi i k_m l_m}{m}}&0
\end{pmatrix} \otimes \begin{pmatrix}
e^{\frac{2 \pi i k_n l_n}{n}}&0\\
0&e^{-\frac{2 \pi i k_n l_n}{n}}
\end{pmatrix}$    \\ $ = \left(
\begin{array}{cccc}
 0 & 0 & e^{\frac{2 i k_m \pi  l_m}{m}+\frac{2 i k_n  l_n  \pi }{n}} & 0 \\
 0 & 0 & 0 & e^{\frac{2 i k_m l_m \pi }{m}-\frac{2 i k_n  l_n  \pi }{n}} \\
 e^{\frac{2 i k_n  l_n  \pi }{n}-\frac{2 i k_m l_m \pi }{m}} & 0 & 0 & 0 \\
 0 & e^{-\frac{2 i k_m \pi  l_m}{m}-\frac{2 i k_n  l_n  \pi }{n}} & 0 & 0 \\
\end{array}
\right)$. So, the character is $\chi_{48} =0$.
  \item $\rho_{49} = \rho_{1,1, m} \otimes \rho_{3,2, n} = 1 \times \begin{pmatrix}
0&1\\
1&0
\end{pmatrix} = \begin{pmatrix}
0&1\\
1&0
\end{pmatrix}$. So, the character is $\chi_{49} =0$.
  \item $\rho_{50} = \rho_{2,1, m} \otimes \rho_{3,2, n} = 1 \times \begin{pmatrix}
0&1\\
1&0
\end{pmatrix}  \text{ or } -1 \times \begin{pmatrix}
0&1\\
1&0
\end{pmatrix} $ $ = \pm \begin{pmatrix}
0&1\\
1&0
\end{pmatrix}$. So, the character is $\chi_{50} =0$.
  \item $\rho_{51} = \rho_{3,1, m} \otimes \rho_{3,2, n} = 1 \times \begin{pmatrix}
0&1\\
1&0
\end{pmatrix} \text{ or } -1 \times \begin{pmatrix}
0&1\\
1&0
\end{pmatrix} $ $ = \pm \begin{pmatrix}
0&1\\
1&0
\end{pmatrix}$. So, the character is $\chi_{51} =0$.
  \item $\rho_{52} = \rho_{4,1, m} \otimes \rho_{3,2, n} = 1 \times \begin{pmatrix}
0&1\\
1&0
\end{pmatrix} \text{ or } -1 \times \begin{pmatrix}
0&1\\
1&0
\end{pmatrix} $ $ = \pm \begin{pmatrix}
0&1\\
1&0
\end{pmatrix}$. So, the character is $\chi_{52} =0$.
  \item $\rho_{53} = \rho_{1,2, m} \otimes \rho_{3,2, n} = \begin{pmatrix}
e^{\frac{2 \pi i k_m}{m}}&0\\
0&e^{-\frac{2 \pi i k_m}{m}}
\end{pmatrix} \otimes \begin{pmatrix}
0&1\\
1&0
\end{pmatrix}$ $ = \left(
\begin{array}{cccc}
 0 & e^{\frac{2 i k_m \pi }{m}} & 0 & 0 \\
 e^{\frac{2 i k_m \pi }{m}} & 0 & 0 & 0 \\
 0 & 0 & 0 & e^{-\frac{1}{m} 2 i k_4 \pi } \\
 0 & 0 & e^{-\frac{1}{m} 2 i k_4 \pi } & 0 \\
\end{array}
\right)$. So, the character is $\chi_{53} =0$.
  \item $\rho_{54} = \rho_{2,2, m} \otimes \rho_{3,2, n} = \begin{pmatrix}
e^{\frac{2 \pi i k_m l_m}{m}}&0\\
0&e^{-\frac{2 \pi i k_m l_m}{m}}
\end{pmatrix} \otimes \begin{pmatrix}
0&1\\
1&0
\end{pmatrix} $ \\ $ = \left(
\begin{array}{cccc}
 0 & e^{\frac{1}{m} 2 i k_m l_m \pi } & 0 & 0 \\
 e^{\frac{1}{m} 2 i k_m l_m \pi } & 0 & 0 & 0 \\
 0 & 0 & 0 & e^{-\frac{1}{m} 2 i k_m l_m \pi } \\
 0 & 0 & e^{-\frac{1}{m} 2 i k_m l_m \pi } & 0 \\
\end{array}
\right)$. So, the character is $\chi_{54} =0$.
  \item $\rho_{55} = \rho_{3,2, m} \otimes \rho_{3,2, n} = \begin{pmatrix}
0&1\\
1&0
\end{pmatrix}  \otimes \begin{pmatrix}
0&1\\
1&0
\end{pmatrix} = \left(
\begin{array}{cccc}
 0 & 0 & 0 & 1 \\
 0 & 0 & 1 & 0 \\
 0 & 1 & 0 & 0 \\
 1 & 0 & 0 & 0 \\
\end{array}
\right)$. So, the character is $\chi_{55} =0$.
  \item $\rho_{56} = \rho_{4,2, m} \otimes \rho_{3,2, n} = \begin{pmatrix}
0&e^{\frac{2 \pi i k_m l_m}{m}}\\
e^{-\frac{2 \pi i k_m l_m}{m}}&0
\end{pmatrix}  \otimes \begin{pmatrix}
0&1\\
1&0
\end{pmatrix}$      \\ $ = \left(
\begin{array}{cccc}
 0 & 0 & 0 & e^{\frac{1}{m} 2 i k_m l_m \pi } \\
 0 & 0 & e^{\frac{1}{m} 2 i k_m l_m \pi } & 0 \\
 0 & e^{-\frac{1}{m} 2 i k_m l_m \pi } & 0 & 0 \\
 e^{-\frac{1}{m} 2 i k_m l_m \pi } & 0 & 0 & 0 \\
\end{array}
\right)$. So, the character is $\chi_{56} = 0$.
  \item $\rho_{57} = \rho_{1,1, m} \otimes \rho_{4,2, n} = 1 \times \begin{pmatrix}
0&e^{\frac{2 \pi i k_n l_n}{n}}\\
e^{-\frac{2 \pi i k_n l_n}{n}}&0
\end{pmatrix}$ $ = \begin{pmatrix}
0&e^{\frac{2 \pi i k_n l_n}{n}}\\
e^{-\frac{2 \pi i k_n l_n}{n}}&0
\end{pmatrix}$. So, the character is $\chi_{57} =0$.
  \item $\rho_{58} = \rho_{2,1, m} \otimes \rho_{4,2, n} = 1 \times \begin{pmatrix}
0&e^{\frac{2 \pi i k_n l_n}{n}}\\
e^{-\frac{2 \pi i k_n l_n}{n}}&0
\end{pmatrix} \text{ or } -1 \times  \begin{pmatrix}
0&e^{\frac{2 \pi i k_n l_n}{n}}\\
e^{-\frac{2 \pi i k_n l_n}{n}}&0
\end{pmatrix}$ \\$ = \pm \begin{pmatrix}
0&e^{\frac{2 \pi i k_n l_n}{n}}\\
e^{-\frac{2 \pi i k_n l_n}{n}}&0
\end{pmatrix}$. So, the character is $\chi_{58} =0$.
  \item $\rho_{59} = \rho_{3,1, m} \otimes \rho_{4,2, n} = 1 \times \begin{pmatrix}
0&e^{\frac{2 \pi i k_n l_n}{n}}\\
e^{-\frac{2 \pi i k_n l_n}{n}}&0
\end{pmatrix} \text{ or } -1 \times \begin{pmatrix}
0&e^{\frac{2 \pi i k_n l_n}{n}}\\
e^{-\frac{2 \pi i k_n l_n}{n}}&0
\end{pmatrix} $\\ $ = \pm \begin{pmatrix}
0&e^{\frac{2 \pi i k_n l_n}{n}}\\
e^{-\frac{2 \pi i k_n l_n}{n}}&0
\end{pmatrix}$. So, the character is $\chi_{59} =0$.
  \item $\rho_{60} = \rho_{4,1, m} \otimes \rho_{4,2, n} = 1 \times \begin{pmatrix}
0&e^{\frac{2 \pi i k_n l_n}{n}}\\
e^{-\frac{2 \pi i k_n l_n}{n}}&0
\end{pmatrix} \text{ or } -1 \times \begin{pmatrix}
0&e^{\frac{2 \pi i k_n l_n}{n}}\\
e^{-\frac{2 \pi i k_n l_n}{n}}&0
\end{pmatrix} $ \\$ = \pm \begin{pmatrix}
0&e^{\frac{2 \pi i k_n l_n}{n}}\\
e^{-\frac{2 \pi i k_n l_n}{n}}&0
\end{pmatrix}$. So, the character is $\chi_{60} =0$.
  \item $\rho_{61} = \rho_{1,2, m} \otimes \rho_{4,2, n} = \begin{pmatrix}
e^{\frac{2 \pi i k_m}{m}}&0\\
0&e^{-\frac{2 \pi i k_m}{m}}
\end{pmatrix} \otimes \begin{pmatrix}
0&e^{\frac{2 \pi i k_n l_n}{n}}\\
e^{-\frac{2 \pi i k_n l_n}{n}}&0
\end{pmatrix}$ \\ $ = \left(
\begin{array}{cccc}
 0 & e^{\frac{2 i \pi  k_m}{m}+\frac{2 i k_n  l_n  \pi }{n}} & 0 & 0 \\
 e^{\frac{2 i k_m \pi }{m}-\frac{2 i k_n  l_n  \pi }{n}} & 0 & 0 & 0 \\
 0 & 0 & 0 & e^{\frac{2 i k_n  l_n  \pi }{n}-\frac{2 i k_m \pi }{m}} \\
 0 & 0 & e^{-\frac{2 i \pi  k_m}{m}-\frac{2 i k_n  l_n  \pi }{n}} & 0 \\
\end{array}
\right)$. So, the character is $\chi_{61} =0$.
  \item $\rho_{62} = \rho_{2,2, m} \otimes \rho_{4,2, n} = \begin{pmatrix}
e^{\frac{2 \pi i k_m l_m}{m}}&0\\
0&e^{-\frac{2 \pi i k_m l_m}{m}}
\end{pmatrix} \otimes \begin{pmatrix}
0&e^{\frac{2 \pi i k_n l_n}{n}}\\
e^{-\frac{2 \pi i k_n l_n}{n}}&0
\end{pmatrix}$ \\ $ = \left(
\begin{array}{cccc}
 0 & e^{\frac{2 i k_m \pi  l_m}{m}+\frac{2 i k_n  l_n  \pi }{n}} & 0 & 0 \\
 e^{\frac{2 i k_m l_m \pi }{m}-\frac{2 i k_n  l_n  \pi }{n}} & 0 & 0 & 0 \\
 0 & 0 & 0 & e^{\frac{2 i k_n  l_n  \pi }{n}-\frac{2 i k_m l_m \pi }{m}} \\
 0 & 0 & e^{-\frac{2 i k_m \pi  l_m}{m}-\frac{2 i k_n  l_n  \pi }{n}} & 0 \\
\end{array}
\right)$. So, the character is $\chi_{62} =0$.
  \item $\rho_{63} = \rho_{3,2, m} \otimes \rho_{4,2, n} = \begin{pmatrix}
0&1\\
1&0
\end{pmatrix} \otimes \begin{pmatrix}
0&e^{\frac{2 \pi i k_n l_n}{n}}\\
e^{-\frac{2 \pi i k_n l_n}{n}}&0
\end{pmatrix}$   \\ $ = \left(
\begin{array}{cccc}
 0 & 0 & 0 & e^{\frac{1}{n} 2 i k_n l_n \pi } \\
 0 & 0 & e^{-\frac{1}{n} 2 i k_n l_n \pi } & 0 \\
 0 & e^{\frac{1}{n} 2 i k_n l_n \pi } & 0 & 0 \\
 e^{-\frac{1}{n} 2 i k_n l_n \pi } & 0 & 0 & 0 \\
\end{array}
\right)$. So, the character is $\chi_{63} =0$.
  \item $\rho_{64} = \rho_{4,2, m} \otimes \rho_{4,2, n} = \begin{pmatrix}
0&e^{\frac{2 \pi i k_m l_m}{m}}\\
e^{-\frac{2 \pi i k_m l_m}{m}}&0
\end{pmatrix} \otimes \begin{pmatrix}
0&e^{\frac{2 \pi i k_n l_n}{n}}\\
e^{-\frac{2 \pi i k_n l_n}{n}}&0
\end{pmatrix}$  \\ = $ \left(
\begin{array}{cccc}
 0 & 0 & 0 & e^{\frac{2 i k_m \pi  l_m}{m}+\frac{2 i k_n  l_n  \pi }{n}} \\
 0 & 0 & e^{\frac{2 i k_m l_m \pi }{m}-\frac{2 i k_n  l_n  \pi }{n}} & 0 \\
 0 & e^{\frac{2 i k_n  l_n  \pi }{n}-\frac{2 i k_m l_m \pi }{m}} & 0 & 0 \\
 e^{-\frac{2 i k_m \pi  l_m}{m}-\frac{2 i k_n  l_n  \pi }{n}} & 0 & 0 & 0 \\
\end{array}
\right)$.\\ So, the character is $\chi_{64} =0$.        
\end{enumerate}

\end{document}